\documentclass[journal,twoside,web]{ieeecolor}

\usepackage{generic}
\usepackage{graphicx,algorithm,algcompatible,amsmath,amsfonts,verbatim,mathtools,enumerate,bm,hyperref,dsfont,pgfplots,subcaption}
\usepackage{algpseudocode}
\usepackage{tikz}

\def\BibTeX{{\rm B\kern-.05em{\sc i\kern-.025em b}\kern-.08em
    T\kern-.1667em\lower.7ex\hbox{E}\kern-.125emX}}
\newcommand{\N}{{\mathbb N}}
\newcommand{\1}{{\mathds{1}}}

\newcommand{\mc}{\mathcal}

\newcommand{\ul}{\underline}
\newcommand{\ol}{\overline}
\DeclareMathOperator{\E}{\mathbb{E}}

\newcommand{\Expect}{\mathop{\bf E{}}}

\newcommand{\argmin}{\mathop{\rm argmin}}
\newcommand{\argmax}{\mathop{\rm argmax}}

\newcommand{\mnorm}[1]{{\left\vert\kern-0.25ex\left\vert\kern-0.25ex\left\vert #1 
    \right\vert\kern-0.25ex\right\vert\kern-0.25ex\right\vert}}

\newtheorem{definition}{Definition} 
\newtheorem{theorem}{Theorem}
\newtheorem{lemma}{Lemma}

\newtheorem{remark}{Remark}
\newtheorem{proposition}{Proposition}

\newtheorem{example}{Example}

\begin{document}
\title{Heterogeneous Multi-Agent Task-Assignment with Uncertain Execution Times and Preferences
}
\author{Qinshuang Wei, \IEEEmembership{Member, IEEE}, Vaibhav Srivastava, \IEEEmembership{Senior Member, IEEE}, Vijay Gupta, \IEEEmembership{Fellow, IEEE}
\thanks{Q. Wei, and V. Gupta are with the Elmore Family School of Electrical and Computer Engineering, Purdue University, West Lafayette, IN 47907, USA (emails: wei473@purdue.edu, gupta869@purdue.edu). V. Srivastava is with the Electrical and Computer Engineering, Michigan State University, East Lansing, MI 48824 (email: vaibhav@egr.msu.edu).
This work was supported in part by the ONR award N00014-22-1-2813. }
}


\maketitle

\begin{abstract}
While sequential task assignment for a single agent has been widely studied, such problems in a multi-agent setting, where the agents have heterogeneous task preferences or capabilities, remain less well-characterized. We study a multi-agent task assignment problem where a central planner assigns recurring tasks to multiple members of a team over a finite time horizon. For any given task, the members have heterogeneous capabilities in terms of task completion times, task resource consumption (which can model variables such as energy or attention), and preferences in terms of the rewards they collect upon task completion. We assume that the reward, execution time, and resource consumption for each member to complete any task are stochastic with unknown distributions. The goal of the planner is to maximize the total expected reward that the team receives over the problem horizon while ensuring that the resource consumption required for any assigned task is within the capability of the agent. We propose and analyze a bandit algorithm for this problem. Since the bandit algorithm relies on solving an optimal task assignment problem repeatedly, we analyze the achievable regret in two cases: when we can solve the optimal task assignment exactly and when we can solve it only approximately. 
\end{abstract}

\section{Introduction}
\label{sec:intro}


A multi-agent task assignment problem aims to maximize the reward from task completion by assigning tasks to agents. We consider a setting where tasks reoccur and agents  may have heterogeneous capabilities and preferences. However, the planner who assigns the tasks to the agents may not be aware of the capabilities and preferences of the agents and has to learn this information over time based on the results from previous task assignments and executions. As an example, we can consider crowdsourcing in which the planner strives to optimize task assignments while considering the realized utility~\cite{zhang2015bandit, zhang2015task,ul2016efficient}, worker reliability~\cite{rangi2018multi}, or preferences~\cite{hassan2014multi,zhao2019preference,zhao2020preference} within a limited time or budget~\cite{rangi2018multi,tran2014efficient}. 


The main challenge in such multi-agent task assignment is that agents have heterogeneous and a priori unknown (to the planner) preferences in terms of the reward they obtain from completing a particular task, characteristics in terms of the time they spend on a task, and resources that they need to consume to successfully complete a task. 
We study a task assignment problem in an online multi-agent setting where the tasks recur over time upon completion. 
A central planner assigns tasks to the agents. The reward, completion time, and resource needed by each agent for any task follow distributions that are unknown to the planner. Assigning tasks to agents who do not have sufficient resources to complete the tasks incurs no reward and instead generates a penalty that is counted separately. The goal of the planner is to maximize the expected total reward that all the agents together can receive over a given finite time horizon from the completed tasks. 
We formulate and solve this problem in a bandit framework to learn the optimal sequential task assignment over the problem horizon.



Although existing literature has studied sequential task assignment problems with uncertainty in rewards, task completion times, and task resource consumption individually, the combination of all these features has not been studied so far. When each agent can execute an arbitrary number of tasks at the same time, the problem can be cast as a combinatorial multi-armed bandit (CMAB) in which the planner pulls multiple arms simultaneously. Here, \cite{chen2013combinatorial} provides a framework that minimizes an approximate regret (with respect to a non-optimal benchmark),~\cite{dong2021efficient} considers the situation when the tasks and agents have uncertain waiting times, and~\cite{kveton2015tight, kveton2014matroid} propose computationally efficient algorithms for stochastic combinatorial semi-bandits. However, unlike this line of work, we additionally consider that the agents consume a stochastic amount of resources to execute any task that precludes the basic assumption made in such works that an arbitrary number of tasks can be assigned to the same agent and brings in a combinatorial aspect to the problem.

Similarly, the requirement of an uncertain amount of resources for the agents to execute any task in a multi-armed bandit (MAB) formulation has also been considered. For instance,~\cite{ding2013multi, cayci2020budget} study the problem when the planner needs to ensure that the cost for pulling any arm remains within a given budget. The studies~\cite{kagrecha2023constrained, amani2019linear} explore a constrained MAB in settings that are compatible with uncertain resource constraints.~\cite{zhou2018budget} considers the case when a fixed number of tasks that each incur an uncertain cost have to be selected at each time when budget constraints are present at each agent.
However, none of the works in this direction consider stochastic and unknown task-completion times as we do, which makes the problem of online task assignment much more difficult since uncertainty in completion time results in any task blocking an agent either fully or partially. 

Some works have also separately explored unknown completion time in the context of blocking bandits. The problem setup in~\cite{ito2024bandit, gyorgy2007continuous} considers stochastic completion times for an agent to complete tasks. While~\cite{gyorgy2007continuous} does not consider multi-tasking for the agent,~\cite{ito2024bandit} considers multi-tasking but only for a single agent. Similarly, the works~\cite{basu2019blocking,basu2021contextual,bishop2020adversarial} have considered blocking bandits, but with fixed blocking times for task completion. Recently,~\cite{atsidakou2021combinatorial} extended blocking bandits to consider uncertain blocking times and~\cite{papadigenopoulos2021recurrent} characterized the regret with respect to a sub-optimal benchmark. However, these works do not consider that an agent may require an uncertain amount of resources to execute any task, which complicates the problem for the planner.



Our main contribution is to formulate, solve, and analyze a sequential multi-agent task assignment problem with heterogeneous agents and recurring tasks in a bandit framework, where the task-completion times, rewards obtained from completing the tasks, and the resources required by the agents to complete any task are all stochastic and follow distributions that are unknown to the planner. The task assignment algorithm we propose is not limited to a particular number of agents and tasks. It utilizes a routine that computes the task assignments given the information known to the planner. Note that an optimal task assignment problem is usually computationally complex, even if the distributions of the various variables are known exactly. For problems of small size, this routine could be one that calculates the optimal task assignment, while for larger problems, the routine can result in approximately optimal task assignments. Our algorithm is robust to such choices. 
We characterize both the regret of our proposed sequential task assignment algorithm (where the reward is counted only over tasks that can be completed given the resource constraints of the agents) and the penalty that the planner incurs due to assigning tasks to agents beyond their resource capacity (the so-called violation penalty). The regret is with respect to a benchmark that is appropriate in the sense that it relies on the optimal or sub-optimal task assignments, depending on the optimal or sub-optimal routine used in our algorithm. 
We show that both the violation penalty and the regret can be upper bounded as $O(\ln T)$, where $T$ is the time horizon, demonstrating the effectiveness of the task assignment algorithm. We validate the proposed algorithm through a case study with three teams of different sizes.

Considering all three sources of uncertainty---rewards, task completion times, and task resource consumptions---together requires solving many technical challenges. First, in the analysis of the violation penalty, the planner needs to assign the tasks based on an estimated lower bound of the resource consumed by an agent for multiple tasks. Note that this lower bound is not simply a summation of the bounds for individual tasks as in traditional constrained MAB studies (e.g., ~\cite{kagrecha2023constrained}), rather a maximum over these individual terms. 
Second, in the regret analysis, as compared to state-of-the-art MAB studies (e.g.,~\cite{ito2024bandit}), on top of the update of the estimated resource consumption, we also need to estimate (and remove) the rewards that the agents obtain from tasks that consume more resources than are available at any agent. The coupling between these two terms makes the regret computation more involved.
Finally, in the regret analysis, the task assignments in every round must account for the fact that we may only obtain an approximately optimal task assignment due to the computational complexity of the problem, and that we have to select feasible task assignments based on our current estimation of resource constraints at the agents. Neglecting this fact in the regret calculation can result in assignments that amplify the approximation error.


The paper is organized as follows. We formulate the sequential task assignment problem in Section~\ref{sec:formulation}, and then propose a bandit algorithm for task assignments in Section~\ref{sec:bandit_alg}, followed by theoretical analysis of the regrets and constraint violation penalty incurred by the algorithm in Section~\ref{sec:thm}. Finally, we demonstrate our algorithm with two case studies in Section~\ref{sec:numerical}.


\emph{Notation:}
The $(i,j)$-th entry of a matrix $M$ is denoted by $M_{ij}$. For two vectors $a,b$ (resp. matrices $A,B$)  of the same size, $a \leq b$ (resp. $A\leq B$) means that each entry of $a$ (resp. $A$) is less or equal than the corresponding entry of $b$ (resp. $B$). Given two $n\times m$ matrices $A,B$, we let $ C = A \odot B$ represent the dot product of two matrices, i.e., $C_{ij} = A_{ij} B_{ij}$. Further, we define the function $\Sigma(A)$ that sums up all entries in the matrix $A$.  We let $\mathbf{1}_N$ (resp. $\mathbf{0}_N$) represent the $N$-dimensional vector of all ones (resp. all zeros), and let $\mathbf{I}_N$ represent the $N \times N$ identity matrix. We denote the complement of a set $S$ as $S^c$. $\1[x]$ denotes the indicator function that is 1 if $x$ is true and 0 otherwise.
\section{Problem Formulation}
\label{sec:formulation}



Consider a team consisting of several agents that need to execute a set of tasks. Each task needs a certain amount of resources from an agent to complete, which can correspond to quantities such as energy or attention from the agent. Both the resource consumed by an agent for a task and the time that an agent takes to complete a task can vary among different agents for a given task, among different tasks for the same agent, and in every recurrence of the same task assigned to the same agent. Each agent has a resource constraint in terms of the total amount of resources it can devote at any time, which can vary among the agents. After completion of each task, the team receives a reward, which can vary with the task and the agent that completed that task. Each task becomes available anew after completion to be assigned again. The planner can assign a task to an agent if both the task and agent are available. The aim is to maximize the expected total team reward over a given time horizon.

More formally, the problem operates in discrete time steps $t=1,\cdots,T$ for a given horizon $T$.  We denote the set of all tasks by $\mc T := (v_1,v_2,\ldots,v_{N})$, where ${N}$ is the total number of tasks. Only one copy of the task is active at any time. However, once the task is finished by an agent, another copy becomes available to be assigned. For simplicity, we assume that the tasks are finished, regenerated, or assigned at the beginning of any time step. Denote the set of all agents by $\mc M := \{\eta_1, \eta_2, \ldots, \eta_{M}\}$, where $M$ is the total number of agents or team members. For notational ease, we define the index sets $I_\mc T := \{1,2,\ldots, N\}$ and $I_\mc M := \{1,2,\ldots, M\}$. For simplicity, we call the agent $\eta_m$ as agent $m$, and the task $v_i$ as task $i$. We can summarize all the {\em task assignments} at any time $t$ through an $N\times M$ binary matrix $a(t),$ whose $(i,m)$-th entry denotes whether the task $i$ has been assigned to agent $m$ (if the entry $a_{im}(t)=1$) or not (if the entry $a_{im}(t)=0$). Note that the task assignment may change over time.  Define the set of all tasks assigned to agent $m$ at time $t$ as $\mc T_m(t) := \{i \in I_\mc T \mid a_{im}(t)=1 \}.$  

Each agent $m$ has a constraint $L_m \geq 0$ in terms of resources it has at its disposal. Define $L := (L_1,L_2,\ldots,L_{M}).$ Once a task ${i}$ is assigned to an agent $m$ at time $t$, the agent must devote a stochastic amount of resource to complete the task at every step till the task completion. We denote by $f_{im}(t)$ the amount of resources that agent $m$ allocate to the task $i$ at time $t$. 
For simplicity and without loss of generality, we scale the amount of resources such that the probability distribution from which $f_{im}(t)$ is sampled for every $i$, $m$, and $t$ has the support $[0, 1].$ We also assume that this distribution has a time-invariant mean $ \bar{f}_{im}.$ 
At each round $t$, a feasible task assignment $a(t)$ is defined below.

\begin{definition}[Feasible Task Assignment]
\label{def:feasible_A}
    A task assignment $a(t)$ is \emph{feasible} if it satisfies the following two properties:
    \begin{enumerate}
        \item For all $m \in I_\mc M$, the agent $m$ can simultaneously execute the set of tasks assigned to her given the resources at her disposal. Given that the amount of resources assigned to a task is stochastically sampled, we impose the expected constraint on $\mc T_{m}(t)$ that $\sum_{i\in I_\mc T} \bar{f}_{im}  a_{im}(t)\leq L_m.$ 
        \item Each task is assigned to at most one agent at the same time. In other words, $\sum_{m\in I_\mc M} a_{im}(t)\leq 1$. 
    \end{enumerate}
\end{definition}
We denote the set of all {\em feasible} task assignments as $\mc A \subseteq \{0,1\}^{N \times M}$. Note that $\mc A$ is closed under inclusion, that is, if $a \in \mc A$ and $b\leq a$, then $b \in \mc A$. Throughout the paper, we assume that $\mc A \neq \emptyset$. The set of all {\em possible} task assignments is defined as the set of task assignments that satisfy the second condition in Definition~\ref{def:feasible_A} and denoted by $\mc A^0 \subseteq \{0,1\}^{N \times M}.$ We denote the index sets of all feasible task assignments by $I_\mc A$ and that of all possible task assignments by $I_{\mc A^0}$. 

Thus, the timeline of the problem is as follows. For a given time horizon $T \in \N$, at each round $t = 1, \cdots, T$, the planner observes the unassigned tasks and the current task assignment, and chooses a feasible task assignment $a(t) \in \mc A$. 
Each agent $m$ needs to spend time $c_{im}(t)$ to execute the tasks assigned to it at round $t$. We assume that whenever task $i$ is assigned to agent $m$, the time $c_{im}(t)$ is sampled from a discrete distribution with support over the set $\{C_l,C_l+1,\cdots, C_u\}$, where $C_l,C_u \in \mc \N$ and $1\leq C_l \leq C_u$. Thus, the agent $m$ completes the task $i$ assigned to it at time $t$ at time $t+c_{im}(t)$. During this time interval that the agent $m$ is executing the task $i$, it allocates an amount of resources $f_{im}(t) \in [0, 1]$ to the task at every time $t$. Finally, upon completion, the agent (equivalently, the team or the planner) receives a reward $r_{im}(t)$. Once again, we assume that the reward is sampled from a distribution with support over the interval $[0,1]$.  Once the task is completed at time $t+c_{im}(t)$, it becomes available for reassignment. 

For pedagogical ease, we assume that the allocation of resources $f_{im}(t)$, the amount of time needed to complete a task $c_{im}(t)$ and the reward $r_{im}(t)$ are chosen from their respective distributions in an independent and identically distributed manner across time, agents, and tasks. However, we note that our arguments below can easily be generalized to consider the completion times and rewards to be correlated at the cost of extra notation. We define $\bar r_{im},\bar c_{im},\bar f_{im}$ to be the expectations of $r_{im}(t), c_{im}(t), f_{im}(t),$ respectively. We can collate the completion times $c_{im}(t)$, resource allocations $f_{im}(t)$ and rewards $f_{im}(t)$ active at any time $t$ into the completion time matrix $c(t)$, resource matrix $f(t)$ and the reward matrix $r(t)$ in a straightforward manner. 

Note that at any time $t$, the planner needs to decide the assignment of only the new tasks regenerated and available at time $t$. However, the assignment of any new tasks must respect the constraints imposed by all tasks already assigned but uncompleted. Define the set of all assignments corresponding to tasks that are yet uncompleted at time $t$ by the binary matrix $b(t)\in\mc A$ with $b_{im}(t)$ as its $(i,m)$-th entry. Then, we can define the set $\mc A(t)$ of task assignments available for the newly regenerated tasks at time $t$ through the relations
\begin{equation}
    \label{eq:available_task}
\begin{split}
    &b_{im}(t) = \sum_{s=1}^{t-1} a_{im}(s)\1[s+c_{im}(s)>t] \,, \\
    &\mc A(t) = \{a \in \mc A \mid a+b(t) \in \mc A \}\,.
\end{split}
\end{equation}
Note that in Definition~\ref{def:feasible_A}, we have imposed only an expected constraint since the planner does not know the precise amount of resources (nor the distribution) allocated to the task. Thus, the task assignment at any time may result in violation of the total resources available at some of the agents. We measure this violation through a penalty term of the form 
\begin{equation}
    \label{eq:violation_penalty}
    V_{T}=\mathbb{E}
    \left[\sum_{t=1}^T \sum_{m=1}^M  \max\{0, \bar{f}_{im} (a_{im}(t)+b_{im}(t)) - L_m\}\right].
\end{equation}

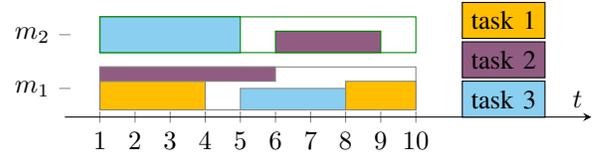
\begin{figure}[!t]
\centering
%
%
\definecolor{task1}{rgb}{1.0, 0.75, 0.0}
\definecolor{task2}{rgb}{0.57, 0.36, 0.51}
\definecolor{task3}{rgb}{0.54, 0.81, 0.94}
\definecolor{m1}{rgb}{0.52, 0.52, 0.51}
\definecolor{m2}{rgb}{0.0, 0.5, 0.0}

\begin{tikzpicture}

\begin{axis}[%
width=7cm,
height=2cm,
at={(0.52in,0.49in)},
scale only axis,
xmin=0,
xmax=15,
ymin= -0.2, ymax=4,
axis lines=middle,
xtick = {0,1,2,3,4,5,6,7,8,9,10},
xlabel = \(t\),
axis y line=left,
y axis line style={opacity=0},
ytick={0.8,2.3},
yticklabels={$m_1$,$m_2$}
]
\draw[m1] (axis cs:1,0.2) rectangle (axis cs:10,1.4);
\draw[m2] (axis cs:1,1.8) rectangle (axis cs:10,2.8);

\fill[task1] (axis cs:1,0.2) rectangle (axis cs:4,0.2+0.8);
\draw[m1] (axis cs:1,0.2) rectangle (axis cs:4,0.2+0.8);

\fill[task2] (axis cs:1,0.2+0.8) rectangle (axis cs:6,1.4);
\draw[m1] (axis cs:1,0.2+0.8) rectangle (axis cs:6,1.4);

\fill[task3] (axis cs:5,0.2) rectangle (axis cs:8,0.2+0.6);
\draw[m1] (axis cs:5,0.2) rectangle (axis cs:8,0.2++0.6);

\fill[task1] (axis cs:8,0.2) rectangle (axis cs:10,0.2+0.8);
\draw[m1] (axis cs:8,0.2) rectangle (axis cs:10,0.2+0.8);

\fill[task3] (axis cs:1,1.8) rectangle (axis cs:5,1.8+1);
\draw[m2] (axis cs:1,1.8) rectangle (axis cs:5,1.8+1);

\fill[task2] (axis cs:6,1.8) rectangle (axis cs:9,1.8+0.6);
\draw[m2] (axis cs:6,1.8) rectangle (axis cs:9,1.8+0.6); 

\node[draw,rectangle,fill =task1,  minimum size=0.5] at(axis cs:12.5,2.7){task 1};

\node[draw,rectangle,fill =task2,  minimum size=0.5] at(axis cs:12.5,1.6){task 2};
\node[draw,rectangle,fill =task3,  minimum size=0.5] at(axis cs:12.5,0.5){task 3};

\end{axis}


\end{tikzpicture}%
\caption{A feasible sequential task assignment in Example~\ref{ex:feasible}. The orange, purple, and blue filled rectangles represent task 1, 2, and 3, respectively. The width of the filled rectangle shows the task completion time, and the height shows the expected resources the member needs to spend on the corresponding task. While the grey and green unfilled rectangles represent the total resources that members 1 and 2 have, respectively.}
\label{fig:feasible}
\end{figure}
\begin{example}
\label{ex:feasible}
    In Fig.~\ref{fig:feasible} we show a feasible sequential task assignment for a team with 2 members $\mc M = \{m_1, m_2\}$ and 3 tasks $\mc V = \{v_1,v_2,v_3\}$, where $L_1= 1.2, L_2=1$, $\bar f_{11} = 0.8$, $\bar f_{21} = \bar f_{22} = 0.6$, $\bar f_{31} = \bar f_{12} = 0.4$, and $\bar f_{32}=1$. The completion times $c_{11}(1) =c_{31}(5) =c_{22}(6)= 3$, $c_{21}(1) = 5$, $c_{11}(8) = 2$, $c_{32}(1) = 4$.
\end{example}

The sequential task assignment problem for the planner is to generate a policy $\pi(t)$ for assigning tasks at every time $t\in[1,T]$ so as to maximize the expected accumulated reward for the team as given by 
\begin{equation}
    \label{eq:expected_accumulated_reward}
   E_{T}= \mathbb{E}\left[\sum_{t=1}^{T} \Sigma\big(a(t) \odot r(t)\big)\mathds{1}\{a(t)+b(t)\in \mc A \}\right].
\end{equation} The policy must be causal and feasible, meaning that it is in the set $\Pi$ of (possibly stochastic) functions that chooses a feasible task assignment $a(t)$ that satisfies~\eqref{eq:available_task} based on information available at the planner prior to round $t$, but not on that after round $t$. Note that the rewards from any tasks that are assigned to an agent without sufficient resources are not counted in $E_{T}$. 

Even if the distribution of the completion times, rewards, and resource consumptions for each agent to perform each task is given, it is still not easy for the central planner to choose the optimal tasks for each agent at each round. We therefore first discuss the optimal sequential task assignment problem when all the distributions information are known to the central planner.

\subsection{Known Distributions of Task Completion Times, Rewards, and Resource Consumption}
\label{sec:known_info}

If the distributions of the completion times, rewards, and resource consumptions are known, the past realizations of these variables are not relevant to the assignment policy. If, further, the completion times for all agents and all tasks are identical, the assignment problem faced by the planner at each time $t$ is a version of the well-studied generalized task assignment problem which is solvable via a mixed integer linear program. With heterogeneous completion times, 
the planner has to assign any new task at time $t$ while considering the assignment $b(t)$ of the tasks  that are still under execution at round $t$ and their respective starting times denoted as $s(b(t))$. For ease of notation, we denote $h(t):=\{b(t),s(b(t))\};$ note that it is sufficient to consider assignment policies $\pi(t)$ such that the assignment at time $t$ depends only on $h(t)$.   
The optimization problem in this case can be written as
\begin{align}
\label{eq:opt_reward}    OPT := & \max_{\substack{\{\pi(1),\cdots,\pi(T)\}\\ \pi(t) \in \Pi }} \E \Big[\sum_{t=1}^T \Sigma( r(t) \odot \pi(t)(h(t))) \Big]\,,
\end{align}
where we note that the history $h(t)$ is itself a function of the previous tasks assignments from policies $\pi(1),\cdots,\pi(t-1).$

Even in the case when the distributions of the completion times, rewards, and resource consumptions are known, solving this maximization problem is computationally challenging. 
However, an upper bound on the objective $OPT$ can be proven along the line of~\cite[Proposition 2.2]{ito2024bandit} by generalizing the setting to a multi-player scenario as in Proposition~\ref{prop:r2q}. Define the `average' or `per-round' reward obtained when agent $m$ completes task $i$ as $q_{im} = \frac{\ol r_{im}}{\ol c_{im}}.$ We also collate these variables in a matrix $q \in [0,1]^{N\times M}.$ 
\begin{proposition}[{{\cite[Proposition 2.2]{ito2024bandit}}}]
\label{prop:r2q}
Consider the problem formulation from Section~\ref{sec:intro} for the case when the distributions of the completion times, rewards, and resource consumptions are known. Define the average reward obtained when agent $m$ completes task $i$ as $q_{im} = \frac{\ol r_{im}}{\ol c_{im}}$ and collate these variables in a matrix $q \in [0,1]^{N\times M}.$  For any feasible task assignment policy, we have 
$$ \E \Big[\sum_{t=1}^T \Sigma( r(t) \odot \pi(t)(h(t))) \Big] 
\leq (T+C_u) \max_{a\in \mc A} \Sigma( q \odot a).$$
\end{proposition}
For pedagogical ease, proofs of all our results are collected in the Appendix.

\subsection{Unknown Distributions of Task Completion Times, Rewards, and Resource Consumption}

If the distributions specifying the reward $r_{im}(t)$, the task processing time $c_{im}(t)$, and the resource consumption $f_{im}(t)$ are unknown to the planner, the past realizations of these variables become important since the planner needs to assign tasks to maximize the team reward while learning these distributions. In the online learning framework that we use, regret with respect to a benchmark algorithm with known distributions is a natural performance metric. Specifically, we compare the expected cumulative reward of our algorithm with the benchmark $\frac{1}{1+\alpha} OPT$ where $\alpha \geq 0$. The choice of $\alpha>0$ encompasses the possibility that the planner may not solve for the optimal task assignment policy due to computational complexity even if the distributions were known.  

Specifically, following the concept of $\frac{1}{1+\alpha}$-regret defined in the literature on blocking bandits~\cite{basu2019blocking} and $\left(\frac{1}{1+\alpha}, 1\right)$-regret defined in the literature on combinatorial bandits~\cite{cesa2012combinatorial}, we evaluate the performance of our learning algorithm through the \emph{$\alpha$-sub-optimal regret} $R^\alpha_T$ defined 
as 
\begin{equation}
\label{eq:regret_subopt}    R^{\alpha}_T =  \frac{1}{1+\alpha}  OPT - E_{T}, 
\end{equation}
where the expected accumulated reward by the algorithm $E_{T}$ is defined in~(\ref{eq:expected_accumulated_reward}). Following the discussion above, when $\alpha=0,$ the regret $R^0_T$ is the \emph{exact regret} in which the planner solves the optimization problem $OPT$ exactly both in the benchmark and at every stage of our bandit algorithm. When $\alpha>0$, the regret $R^\alpha_T$ is the \emph{approximate regret} where the planner solves the optimization problems sub-optimally in both cases. Clearly, this more general formulation allows an apples to apples comparison of our algorithm with an appropriate benchmark. Apart from the regret, the violation penalty $V_{T}$ of any algorithm, as defined in~(\ref{eq:violation_penalty}) also serves as a performance measure.




%
\section{Task-Assignment Algorithm}
\label{sec:bandit_alg}
In this section, we propose a learning algorithm for the problem posed above. Calculating the regret $R_{T}^{\alpha}$ analytically seems computationally intractable. Instead, we upper bound the regret as follows.  
\begin{align}
\nonumber  R_T^{\alpha} &\leq \frac{1}{\alpha+1} (T+C_u) \max_{a\in \mc A} \Sigma(q \odot a) \\
\nonumber  & \qquad - \E\Big[\sum_{t=1}^T \Sigma(r(t)  \odot a(t)) \mathds{1}\{a_t+b_t\in \mc A \}\Big] \\
\nonumber &\leq   \frac{C_u \ol{L}}{(\alpha+1)C_l} + \E\Big[\sum_{t=1}^T  \Sigma \Big(\frac{1}{\alpha+1} q \odot a^*\Big) \\& -\Sigma \Big(r(t)\odot  a(t) \Big)\mathds{1}\{a(t)+b(t)\in \mc A \} \Big]
\label{eq:R_bound_q}\,,
\end{align}
where the first inequality follows from Proposition~\ref{prop:r2q} and in the second inequality, we have defined  $a^*$ as an optimal `per-round' task assignment $a^* \in \argmax_{a \in \mc A} \Sigma(a \odot q)$ and the maximum number of tasks the team can execute simultaneously as $\ol{L}=\max_{a\in \mc A} \sum_{(i,m) \in I_\mc A} a_{im}$ (which we assume is upper bounded by $N$). 

Based on~\eqref{eq:R_bound_q}, our learning algorithm (Algorithm~\ref{alg:bandit}) proceeds by estimating $q$ and $f$. The core idea is to implement a phased-update method combined with a UCB-based and a LCB-based policy update. We divide the time horizon into phases that may possibly be of unequal length. The algorithm updates the task assignment policy at the beginning of each phase, but utilizes the same policy for the rest of the phase.
\begin{algorithm}
\caption{Bandit Algorithm for Task Assigning}\label{alg:bandit}
\begin{algorithmic}[1]
\Require $C_l,C_u$: lower and upper bounds for the
processing time. $B$: length of initialization phase. 
\State Initialization: For all $i \in I_\mc V$ and $m \in I_\mc M$, let member $m$ execute the task $i$ for $B$ times. Record the completing round as $t_1$, and set $T_{im}(t_1) = B$, and $T^f_{im}(t_1)$ as the total rounds for member $m$ to complete the $B$ times of task $i$. \label{line:init} 
\For{$s = 1,2, \cdots, S $}     \label{line:s_start}
    \State Compute $\hat{q}_{im}(t_s)$ using~\eqref{eq:UCB_q} and $d^f_{am}(t_s)$ using~\eqref{eq:LCB_f}.   \label{line:set_q}
        \State Compute $l_s = C_l \min_{(i,m) \in I_{\mc A}} T_{im}(t_s)+ 2C_u$, and set $t_{s+1} = t_s + l_s$.      \label{line:set_length}
    \State $\mc A^\dagger_{s} = \{a\in \mc A^0 \mid \sum_{i \in I_\mc V} \hat{f}_{im}(t_s) a_{im} - d^f_{am}(t) \leq L_m \quad \forall m \}$. \label{line:A_est}
    \If{$\mc A^\dagger_s \neq \emptyset$}
        \State Compute \begin{equation}
            \label{eq:UCB_opt}
            a'_s \in \argmax_{a \in \mc A^\dagger_s} \Sigma(\hat{q}(t_s) \odot a)
            \end{equation} through a given oracle A.
        \label{line:act_select}
    \Else
        \State Compute \begin{equation}
            \label{eq:LCB_infeasible}
            a'_s \in \argmin_{a \in {\mc A}^0}\sum_{m} \max \{ \sum_{i} \hat{f}_{im}(t_s) a_{im} - d^f_{am}(t) ,0 \}
            \end{equation}
    \EndIf
    \For{$t = t_s, t_s+1, \cdots, t_{s+1}-1$}
        \State Compute $b_t$ using \eqref{eq:available_task}.    \label{line:act_avail}
        \State If $b(t) \leq a'_s$, then $a(t) \leftarrow a'_s-b(t)$; else $a(t) \leftarrow \{0\}^{N\times M}$.  \label{line:set_act_start}
        \State If member $m$ executes task $i$ at round $t$, we observe the corresponding resource occupation $f_{t,im}$ and update $\hat{f}_{im}(t)$. Set $T^f_{im}(t+1) = T^f_{im}(t)+1$. \label{line:update_f}
        \For{$i = 1,\cdots,N$}
            \State If member $m$ completes task $i$ at round t and we observe the corresponding reward completion time $r_{t',im}, c_{t',im}$, where $t'$ is the task starting round, update $\hat{r}_{im}(t)$, $\hat{c}_{im}(t)$, $V^c_{im}(t)$. Set $T_{im}(t+1) = T_{im}(t)+1$. \label{line:update}
        \EndFor
    \EndFor
\EndFor  \label{line:s_end}
\end{algorithmic}
\end{algorithm}

Specifically, the first phase (Phase $0$, given in Line~\ref{line:init}) of Algorithm~\ref{alg:bandit} serves as the initialization phase. The task assignment policy in this phase is that each member $m$ completes each task $i$ for $B$ times (where we let $B = \left\lceil 90 \frac{C_u}{C_l} \ln T \right\rceil$ and $T > MNBC_u$). 
We can then upper bound the total regret in phase $0$ by $\mc O(\frac{C_uMNB}{C_l}) = \mc O(\frac{C^2_u}{C^2_l}MN \ln T)$.

In every subsequent phase $s$ prior to the phase that contains $T$, 
the planner follows lines~\ref{line:s_start}--\ref{line:s_end}. We define $S$ as the phase that includes round $T$. We further define $T_{im}(t)$ as the times that agent $m$ completes task $i$ before time $t$, $\hat r_{im}(t)$ and $\hat c_{im}(t)$ as the empirical means of corresponding reward and cost, and $V^c_{t,im}$ as the empirical variance of the corresponding cost. Let $t_s = \sum_{k=1}^{s-1} l_k$ be the time at which phase $s$ begins. 
At the beginning of each phase, we compute the length $l_{s}$ of that phase as 
\begin{equation}
    \label{eq:ls}
    l_s = C_l \min_{(i,m) \in I_{\mc A}} T_{im}(t_s)+ 2C_u.
\end{equation} 
We maintain a UCB-type estimate, $\hat{q}_{im} (t_s)$, of $q_{im}$ 
through~\eqref{eq:UCB_q} below, where 
\begin{equation}
\label{eq:UCB_q}    
\hat q_{im}(t)  = \frac{\min\{1, \hat r_{im}(t)+d^r_{im}(t)\}}{\max\{C_l, \hat c_{im}(t)-d^c_{im}(t)\}},
\end{equation}
where $d^r_{im}(t) = \sqrt{1.5 \ln t /T_{im}(t)}$ and $d^c_{im}(t) = \sqrt{3 V^c_{t,im} \ln t /T_{im}(t)} + 9(C_u-C_l)\ln t/T_{im}(t)$.

Further, we let $\hat{f}_{im}(t)$ be the empirical mean of the resource occupation ${f}_{im}$ at round $t$. We assume that we obtain the information of the resource occupation, $f_{im}(t)$ for any member $m$ executing task $i$ at time $t$, and thus we let $T^f_{im}(t)$ be the number of times that $f_{im}$ gets estimated before round $t$, i.e., the number of rounds that the member $m$ has been executing the task $i$. 
We then set the lower confidence bound (LCB) of the expected resource required for an assignment $a$ at round $t$ as $\sum_{i \in I_\mc V} \hat{f}_{im}(t_s) a_{im} - d^f_{am}(t)$, where
\begin{equation}
\label{eq:LCB_f} 
 d^f_{am}(t)  = \sqrt{\frac{1.5\bar{L}^2 \ln t }{\min_{i:a_{im}=1}{T^f_{im}(t)}}}. 
\end{equation}
We define $d^f_{im}(t)  = \sqrt{\frac{1.5 \ln t }{T^f_{im}(t)}},$ and thus can alternatively write $d^f_{am}(t) = \bar{L}\max_{i:a_{im}=1} d^f_{im}(t)$.
We then estimate the set of feasible tasks in line~\ref{line:A_est} with the LCB estimate of resource occupation.

Having updated these estimates, we then calculate our task assignment policy $a_s'$ for phase $s$ using line~\ref{line:act_select}. This step requires solving the linear optimization problem~\eqref{eq:UCB_opt} via some given oracle A.
During phase $s$, the members will follow the task assignment in line~\ref{line:act_select} after completing all tasks that start in phase $s-1$. Note that each member $m$ repeats the task $i$ such that $a'_{s,im} = 1$ once it is completed.  



\begin{remark}[Oracle A]
    Algorithm~\ref{alg:bandit} is agnostic to the optimization solving oracle A used to solve the problem in Line~\ref{line:act_select}, whether the oracle solves the problem exactly or approximately. The regret guarantees of the algorithm may be different for these cases and we will analyze them under two sets of assumptions. The ability to consider an approximate solution to~\eqref{eq:UCB_opt} is particularly useful since~\eqref{eq:UCB_opt} is equivalent to a mixed-integer optimization problem (MIP) equivalent to a generalized assignment problem (GAP), which is known to be NP-hard in general and computationally efficient algorithms to determine the optimal solution may not exist. There are some algorithms and commercial solvers to determine a sub-optimal solution. For instance, the algorithm in~\cite{cohen2006efficient} guarantees an \emph{$(1+\alpha)$-approximate} solution $a'$, with the guarantee that 
    $$\Sigma(\hat{q}(t_s) \odot a') \geq \frac{1}{1+\alpha} \Sigma(\hat{q}(t_s) \odot \hat{a}^*),$$
    where $\hat{a}^*$ is an optimal solution to~\eqref{eq:UCB_opt}, and the MIP solver, Gurobi~\cite{gurobi}, guarantees an \emph{$(1+\alpha)$-approximate} solution by setting a proper MIP gap. 
\end{remark}
We call Algorithm~\ref{alg:bandit} an \emph{exact algorithm} when the oracle A in Line~\ref{line:act_select} will solve the optimization problem~\eqref{eq:UCB_opt} exactly. To obtain regret guarantees, we also consider the case when oracle A only guarantees an $(1+\alpha)$-approximate solution to~\eqref{eq:UCB_opt}, such as the one in~\cite{cohen2006efficient}. In this case, we call Algorithm~\ref{alg:bandit} an \emph{$(1+\alpha)$-approximate algorithm} or \emph{approximate algorithm}.

\section{Theoretical results}
\label{sec:thm}
In this section, we characterize the regret of the exact and approximate algorithms proposed above. Some parts of our proof techniques follow those in~\cite{ito2024bandit} and~\cite{kagrecha2023constrained}. However, given that~\cite{ito2024bandit} focuses on single-agent task assignment problems with the premise that Algorithm~\ref{alg:bandit} is an exact algorithm and the resource constraint is non-stochastic, and~\cite{kagrecha2023constrained} neither considers stochastic completion times nor a combinatorial setting, our derivation of results remains challenging in the following three aspects. First, when the Algorithm~\ref{alg:bandit} is an exact algorithm, we need to readapt the results in~\cite{ito2024bandit} into a heterogeneous multi-agent setup to provide an upper bound for the sub-optimal regret.
Second, when the Algorithm~\ref{alg:bandit} is an $(1+\alpha)$-approximate algorithm, we need to derive the approximate regret analysis with different techniques, as the task assignments at all rounds are based on that we only obtain an approximate solution to~\eqref{eq:UCB_opt}.
Third, we need to consider the stochastic resource constraints in both analysis of sub-optimal regrets and violation regrets, as even if we assign the tasks that have high per-unit rewards, that assignment may be infeasible, as sometimes $\mc A^\dagger_s \neq \mc A$.

As we set the initialization task completion times as $B = \left\lceil 90 \frac{C_u}{C_l} \ln T \right\rceil$, we have the properties~\eqref{eq:basic_time}--\eqref{eq:berns_ineq} for all $(i,m) \in I_\mc A$, and Lemmas~\ref{lem:phase_length}--\ref{lem:phase_num} below to derive the later analysis of regret bounds.
\begin{equation}
\label{eq:basic_time}
\begin{split}
   & T_{im}(t_s) \geq 90  \frac{C_u}{C_l} \ln T, \quad  l_s \geq 90 C_u \ln T \\
   & 90  C_u \ln T \leq t_1 = \mc O \Big(\frac{C^2_u N M \ln T}{C_l}  \Big)\,.
\end{split}
\end{equation}
\begin{equation}
\label{eq:berns_ineq}
\begin{split}
    \Pr[|\hat r_{im}(t)-\ol{r}_{im}| \geq d^r_{im}(t)] & \leq \frac{2}{t^2}\\
    \Pr[|\hat c_{im}(t)-\ol{c}_{im}| \geq d^c_{im}(t)] & \leq \frac{4}{t^2}\\
    \Pr[q_{im}\leq \hat q_{im}(t)] & \leq 1- \frac{6}{t^2} \,.
\end{split}
\end{equation}

\begin{lemma}[{{\cite[Lemma 3.1]{ito2024bandit}}}]
\label{lem:phase_length}
    For any $s \geq 1$ and $(i,m) \in A'_s$, the phase length $l_s \leq 2 C_u (T_{im}(t_s+1) - T_{im}(t_s))$ and the accumulated execution times $T_{im}(t_{s+1}) \leq 4T_{im}(t_s)$.
\end{lemma}

\begin{lemma}[{{\cite[Lemma 3.2]{ito2024bandit}}}]
\label{lem:phase_num}
    For any $s \geq 1$, there exists $(i,m) \in A'_s$, such that $T_{im}(t_{s+1}) \geq (1+\frac{C_l}{C_u})T_{im}(t_s)$. As a result, for all $s \geq 1$, $t_s \geq C_l (1+\frac{C_l}{C_u})^{\frac{s}{NM}-2}$.
\end{lemma}

Lemma~\ref{lem:phase_num} and that $\ln (1+x) \geq x/2$ for any $x\in [0,1]$ infer that $\ln t_s \geq \ln C_l +  (\frac{s}{NM}-2) \ln (1+\frac{C_l}{C_u}) > \ln C_l + (\frac{s}{NM}-2)\frac{C_l}{2C_u}$. And hence 
\begin{equation}
    \label{eq:S_bound}
    S \leq NM \Big(\frac{2C_u}{C_l} \ln T +2 \Big)+1 = \mc O \Big(\frac{C_u}{C_l}NM\ln T \Big) \,.
\end{equation}

Further, we define an optimal per-round task assignment and its corresponding index set as $A^* = \{(i,m)\in I_\mc A \mid a^*_{im}=1 \}$.
Further, we define the sub-optimality gaps as
\begin{equation}
\label{eq:gap} 
\begin{split}
   & \Delta^\alpha_{a} = \frac{1}{1+\alpha} \sum_{a^*_{im}=1} q_{im} - \sum_{a_{im}=1} q_{im} \quad \forall a\in \mc A\\
   & \Delta^\alpha_{im} = \min_{a\in \mc A\mid a_{im}=1,  \Delta^\alpha_{a} > 0} \Delta^\alpha_a \quad \forall (m,i) \in I_\mc A\\
   & \ul \Delta^\alpha  = \begin{cases}
       \min_{(m,i)\in I_\mc A \backslash A^*} \Delta^0_{im}  \quad \text{ if } \alpha = 0 \\ 
       \min_{(m,i)\in I_\mc A} \Delta^\alpha_{im}  \quad \text{ if } \alpha >0 \,.
   \end{cases}
\end{split}
\end{equation} 
For all $m$ and $a$, we further define the violation gap as below:
\begin{equation}
\begin{split}
\label{eq:inf_gap}
    &\Delta^L_{am} = \max\{\sum_{i}\bar f_{im}(t)a_{im}-L_m, 0 \},\\
    &\Delta^L_{im} = \max\{\bar f_{im}(t)-L_m, 0 \},
    \Delta^L_{a} = \sum_m \Delta^L_{am}.
\end{split}
\end{equation}

Finally, we denote the variance of the processing time $c_{t,im}$ as $\sigma^2_{im} = \Expect[(c_{t,im}-\bar{c}_{im})^2].$ We note that 
\begin{align*}  
\nonumber    \sigma^2_{im} \leq \Expect[(c_{im}(t)-C_l)^2] &\leq (C_u - C_l)\Expect[c_{im}(t)-C_l] \\
 \label{eq:var_c}   & = (C_u - C_l)(\bar{c}_{im}-C_l) \,.
\end{align*}

We first show the performance of the feasible-set estimation in Theorem~\ref{thm:inf_regret_stc_f} below.
\begin{theorem}
\label{thm:inf_regret_stc_f}
If $T > MNBC_u$, the violation penalty $V_T$ defined in~\eqref{eq:violation_penalty} for Algorithm~\ref{alg:bandit} is bounded above by\footnote{Big O notation describes an upper bound on the growth rate of a function, and big Omega notation provides a lower bound.}
\begin{equation}  
\label{eq:inf_bound}
\begin{split}
         V_T & \leq \mc O\Big( \frac{C^2_u}{C_l}\ln T + \frac{\bar{L}^2 \ln T}{(\max_{m}\Delta^L_{am})^2} + \bar{L} \max_{a\in \mc A^c}\Delta^L_{a} \Big) \,.
\end{split}
\end{equation}
\end{theorem}
On the other hand, we adopt~\cite[Theorem 6]{kagrecha2023constrained} to provide a lower bound for the violation penalty. For any joint distribution of rewards, completion times, and resource consumption for all tasks and agents, denoted as $D$, in $\mc D$, for each task assignment $a$ that is not optimal (either sub-optimal or infeasible), we first define that the rewards, completion times, and resource consumption for all tasks that are drawn from the distribution $D$ as $r(D)$, $c(D)$, and $f(D)$. We further let $KL(D,D')$ be the KL divergence between the two distributions $D, D'$. We let 
\begin{equation*}
\begin{split}
\eta_a(D,a^*, L, \mc D) = \inf_{D' \in \mc D} \{ KL(D,D') \mid 
\Sigma\Big(a\odot \frac{\bar{r}(D')}{\bar{c}(D')}\Big) \\ < \Sigma(a^* \odot q)  \,,
\sum_{i} \bar{f}(D')_{im} a_{im} \leq L_m, \, \forall m \} \,.
\end{split}
\end{equation*} 
Then we have
\begin{equation}
    \liminf_{T \to \infty} \frac{\E_{(r,c,f)}[T^f_a(T)]}{\ln T} \geq \frac{1}{\eta_a((r,c,f),a^*, L, \mc D)}\,,
\end{equation}
where $T^f_a(T)$ is the number of rounds the team executes the task assignment $a$. We therefore have
\begin{equation}
    V_T \geq \Omega \Big(\sum_{a \in \mc A^c} \frac{\Delta^L_a \ln T }{\eta_a((r,c,f),a^*, L, \mc D)} \Big )\,,
\end{equation}

We next estimate the performance of the exact and approximate algorithms by analyzing the bounds for both the exact and approximate regrets $R^\alpha_T$ defined as in~\eqref{eq:regret_subopt}. 

\begin{theorem}
\label{thm:sub_optimal_bound}
Given that $T > MNBC_u$ and Algorithm~\ref{alg:bandit} is an $(1+\alpha)$-approximate algorithm for some $\alpha \geq 0$, the regret $R^{\alpha}_T$ defined in~\eqref{eq:regret_subopt} for Algorithm~\ref{alg:bandit} is bounded above by
\begin{equation}  
\label{eq:subopt_bound}
\begin{split}
        R^{\alpha}_T & \leq \mc O\Big(\Big(\frac{1 }{\ul{\Delta}^\alpha } + C_u\Big)\frac{C_u}{C_l^2}NM\ol{L} \ln T \\ 
        & \hspace{2cm} +\sum_{a \in \mc A^c} \frac{\bar{L}^3\ln T}{C_l(\max_{m}\Delta^L_{am})^2} \Big) \,.
\end{split}
\end{equation}
Further, for any distribution of completion times and rewards\footnote{Gap-dependent regret (as in~\eqref{eq:subopt_bound}) refers to regret bounds that explicitly depend on the specific parameters of the underlying distributions (e.g., sub-optimal gaps in~\eqref{eq:gap}), while distribution-free regret here provides bounds that hold uniformly across all possible distributions.}, the regret is bounded as $R^{\alpha}_T = \mc O\Big(\frac{1}{C_l} \sqrt{ C_u N M\ol{L} T\ln T} + \frac{C^2_u}{C_l^2}NM\ol{L} \ln T  +\sum_{a \in \mc A^c} \frac{\bar{L}^3\ln T}{C_l(\max_{m}\Delta^L_{am})^2}\Big)$.
    
Moreover, when Algorithm~\ref{alg:bandit} is an exact algorithm, for any $N,M,T, L$ such that $T = \Omega(C_u)$, there exists a problem instance for which the optimal regret of any task assignment algorithm is lower bounded by $R^{0}_T= \Omega(\frac{1}{C_l} \min \{\sqrt{C_u NM \ol{L} T}, \ol{L}T\})$.
\end{theorem}

\begin{proofsketch}


Given~\eqref{eq:R_bound_q}, we have 
\begin{align}
\nonumber  R_T^{\alpha} & \leq 
\E\Big[\sum_{t=1}^T  \Sigma \Big(\frac{1}{\alpha+1} q \odot a^*\Big)\\
\nonumber & \quad -  \Sigma \Big(r_{t}\odot  a_{t} \Big)\mathds{1}\{a_t+b_t\in \mc A \} \Big] +  \frac{C_u \ol{L}}{(\alpha+1)C_l} \\
\label{eq:R_bound_12} & \leq  \frac{\ol{L} (\E[t_1]+C_u) }{(\alpha+1) C_l} + R^{(1)}_T+ R^{(2)}_T+ R^{(3)}_T \,,
\end{align}

where 
\begin{align}
\nonumber    R^{(1)}_T &= \E\Big[\sum_{s=1}^S \sum_{t=t_s}^{t_{s+1}-1}  \Sigma \Big(\frac{1}{\alpha+1} q \odot a^* - q\odot  a_{s}' \Big) \Big] \\
\nonumber   R^{(2)}_T &=  \E\Big[\sum_{s=1}^S \sum_{t=t_s}^{t_{s+1}-1}  \Sigma \Big(q\odot  a_{s}'- r_{t}\odot  a_{t} \Big) \mathds{1}\{a_t+b_t\in \mc A \} \Big]\\
\nonumber   R^{(3)}_T &=  \E\Big[\sum_{s=1}^S \sum_{t=t_s}^{t_{s+1}-1}  \Sigma \Big(q\odot  a_{s}'\Big) \mathds{1}\{a_t+b_t\in \mc A^c\} \Big]
\end{align}
and we define $t_{S+1} = T+1$ for convenience. 
We then show that 
\begin{align}
\label{eq:R_1_bound}
   R^{(1)}_T &= \mc O \Big(\ol{L} \sum_{ \forall (i,m) \in I_\mc A } \frac{\tilde{C}_{im} \ln T }{\Delta^\alpha_{im}}  \ln T+ \frac{C_u^2}{C_l^2}NM\ol{L}  \Big)\\
\label{eq:R_2_bound}
    R^{(2)}_T &= \mc O\Big(\frac{C_u^2}{C_l^2}NM\ol{L} \ln T \Big) \\\
\label{eq:R_3_bound}
R^{(3)}_T &= \mc O\Big(\frac{C_u^2}{C_l^2}NM\ol{L} \ln T  +\sum_{a \in \mc A^c} \frac{\bar{L}^3\ln T}{C_l(\max_{m}\Delta^L_{am})^2}\Big) \,,
\end{align}
where $\tilde{C}_{im}\leq \mc O\Big(  \frac{C_u}{\ol{c}_{im}C_l} \Big)$.

Lastly, We obtain the distribution-independent regret bound for $R^0_T$ by modifying the analysis of $R^{(1)}_T$. For detailed analysis, please see Appendix~\ref{sec:proof_subopt_regret_bound}. 
\end{proofsketch}

\section{Numerical Study}
\label{sec:numerical}

We demonstrate the proposed bandit heterogeneous task-assigning algorithm on a case study with 2 teams of different sizes and with different tasks. 
We perform the simulations using Matlab, where we use Gurobi with Yalmip toolbox as the computing oracle A for~\eqref{eq:UCB_opt} at the beginning of each phase $s$\footnote{All our code and data used in the case study can be found in https://github.com/QinshuangCoolWei/Heterogeneous-Task-Assignment-with-Uncertain-Execution-Times-and-Preferences.git.}.

The small team has $N=4$ tasks, $M =2$ team members, $L = \{1.5,1.2\}$, the expected reward, completion time, and resource occupation matrices are: 
\begin{equation*}
    \bar{r} = \begin{bmatrix}
    0.525 & 0.45 \\
    0.45 & 0.525 \\
    0.6 & 0.5 \\
    0.5 & 0.7
\end{bmatrix} , \,
\bar{c} = \begin{bmatrix}
    1.5 & 1.5 \\
    1.5 & 1.5 \\
    2 & 2 \\
    2 & 2
\end{bmatrix}, \,
\bar{f} = \begin{bmatrix}
    0.4 & 0.6 \\
    0.6 & 0.5 \\
    0.4 & 0.6 \\
    0.6 & 0.7 
\end{bmatrix}.
\end{equation*}
We set $T = 2\times 10^6$. We run the trial for 50 times. 

When Algorithm~\ref{alg:bandit} is an exact algorithm, 
Fig.~\ref{fig:small_pf_exact} demonstrates the violation penalty $V_T$ and Fig.~\ref{fig:small_R_sub_exact} demonstrates the exact regret $R^{0}_T$. We can observe that both the violation regret and the approximate regret always scale logarithmically with $T$.

Then, when Algorithm~\ref{alg:bandit} is an 2-approximate algorithm ($\alpha = 1$), we demonstrate the violation penalty, $V_T$, in Fig. 3a and the exact regrets, $R^{0}_T$ (blue curve) and $R^{1}_T$ (green curve), in Fig.~\ref{fig:small_pf_approx}. We can observe that the violation penalty still scales logarithmically with $T$, while the approximate regret $R^{1}_T$ (the green curve in Fig.~\ref{fig:small_R_sub_approx}) goes negative and scales linearly with time. The reason for this negative, linear regret is that the approximate algorithm can outperform the benchmark $\frac{1}{2}OPT$, where $OPT$ is the optimal accumulated reward that the team may obtain. 

However, the blue curve in Fig.~\ref{fig:small_R_sub_approx} illustrates that the regret, $R^{0}_T$, still scales logarithmically with $T$, because the approximate algorithm can sometimes perform as good as the exact algorithm, depending on the detailed parameters.

Notice that we also tested $1.1$ and $1.5$-approximate algorithm ($\alpha = 0.1$ and $0.5$) in this small team, which yield similar results as the 2-approximate algorithm does.


\begin{figure}[h!]
    \centering
    \begin{subfigure}[b]{0.45\columnwidth}
        \centering
        \includegraphics[width=\linewidth]{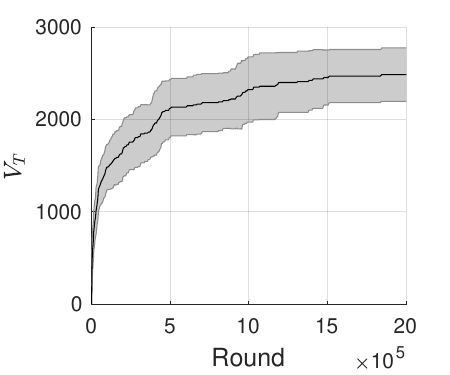}
        \caption{$V_T$ for exact algorithm}
        \label{fig:small_pf_exact}
    \end{subfigure}%
    \hfill
    \begin{subfigure}[b]{0.45\columnwidth}
        \centering
        \includegraphics[width=\linewidth]{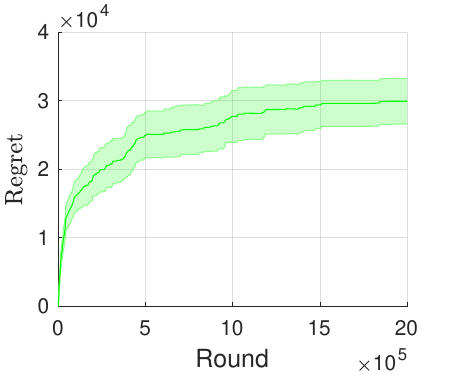}
        \caption{$R^{0}_T$ for exact algorithm}
        \label{fig:small_R_sub_exact}
    \end{subfigure}\\
    \begin{subfigure}[b]{0.45\columnwidth}
        \centering
        \includegraphics[width=\linewidth]{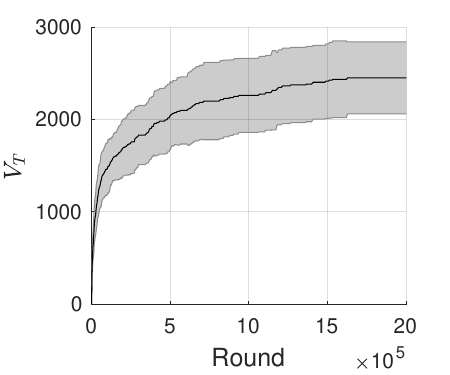}
        \caption{$V_T$ for approximate algorithm}
        \label{fig:small_pf_approx}
    \end{subfigure}%
    \hfill
    \begin{subfigure}[b]{0.45\columnwidth}
        \centering
        \includegraphics[width=\linewidth]{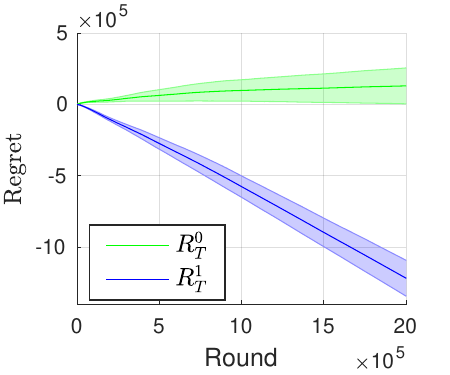}
        \caption{$R^{0}_T$ and $R^{1}_T$ for approximate algorithm}
        \label{fig:small_R_sub_approx}
    \end{subfigure}\\
\caption{Violation Penalties, $V_T$, scale as logarithmic in Fig.~\ref{fig:small_pf_exact} and Fig.~\ref{fig:small_pf_approx}. The exact regrets, $R^{0}_T$, scale as logarithmic in Fig.~\ref{fig:small_R_sub_exact} and Fig.~\ref{fig:small_R_sub_approx} (the green curves). The approximate regret, $R^{1}_T$, scale as negative linear in Fig.~\ref{fig:small_R_sub_approx} (the blue curve). 
}
\label{fig:simulation_small}
\end{figure}

We then simulate the 1.5-approximate algorithm on the large team.
The large team has $N=20$ tasks, $M =5$ team members, and $L_m = 5$ for all $m$. We provide expected rewards, completion times, and resource occupation for the large teams with code given the size of the problem. The 1.5-approximate algorithm can assign the tasks properly to the team at each round in a timely manner. We set $T = 10^7$. All results below are the average of 50 realizations of the same parameters. 

Similar to the results with the small team, in Fig.~\ref{fig:large_pf_approx}, we demonstrate that the violation penalty, $V_T$, is first scaling logarithmically with $T$ and then remain constant, and in Fig.~\ref{fig:large_R_sub_approx}, we show that the approximate regret $R^{0.5}_T$ first increase and then goes negative and scales linearly with time, and thus is outperforming the benchmark, $\frac{1}{1.5}OPT$.
\begin{figure}[h!]
    \centering
    \begin{subfigure}[b]{0.45\columnwidth}
        \centering
        \includegraphics[width=\linewidth]{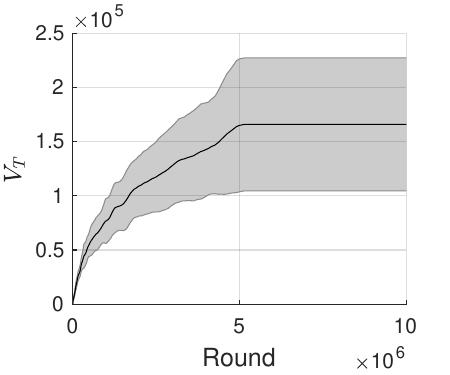}
        \caption{$V_T$ for approximate algorithm}
        \label{fig:large_pf_approx}
    \end{subfigure}%
    \hfill
    \begin{subfigure}[b]{0.45\columnwidth}
        \centering
        \includegraphics[width=\linewidth]{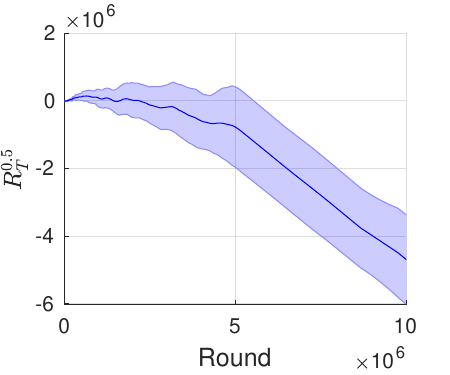}
        \caption{$R^{0.5}_T$ for approximate algorithm}
        \label{fig:large_R_sub_approx}
    \end{subfigure}\\
\caption{Violation penalty, $V_T$, scales as logarithmic in Fig.~\ref{fig:large_pf_approx}. Approximate regret, $R^{0.5}_T$, scales as negative linear in Fig.~\ref{fig:large_R_sub_approx}. 
}
\label{fig:simulation_large}
\end{figure}

\section{Conclusion}
\label{sec:conclusion}
We introduced the heterogeneous sequential task assignment problem with uncertain resource consumption in a multi-agent setting, and studied an efficient task assignment algorithm for both small-scale and large-scale problems where computing the exact optimal solution at each round can be intractable. We theoretically analyzed the violation penalties and the exact and approximate regret bounds for both the exact and the approximate algorithms, respectively, and demonstrated the effectiveness of the algorithm via case studies with teams of 2 sizes.

The proposed task assignment algorithm adapts to heterogeneous teammate preferences and abilities by modeling them as fixed distributions. However, in practice, these attributes can evolve over time, as individual proficiency and interest in tasks may fluctuate based on time, rewards, etc. Therefore, a future direction is to incorporate these evolving preferences and abilities into the task assignment process to plan more adaptively at all times.




\bibliographystyle{IEEEtran}
\bibliography{reference}

\begin{appendices}

\section{Proof for Proposition~\ref{prop:r2q}, Lemma~\ref{lem:phase_length} and~\ref{lem:phase_num}}
\label{sec:proof_r2q}

\begin{proof} 
Our multi-member task assignment setup with $N$ tasks and $M$ members with at most $\ol L$ simultaneous tasks is equivalent to that for a single agent task assignment with $NM$ tasks and $\ol L$ maximal simultaneous tasks. As a result, we can directly apply the proof for~\cite[Proposition 2.2]{ito2024bandit} to Proposition~\ref{prop:r2q} here, and apply~\cite[Lemma 3.1 and 3.2]{ito2024bandit} to Lemma~\ref{lem:phase_length} and~\ref{lem:phase_num} here.
\end{proof}

\section{Proof for Theorem~\ref{thm:inf_regret_stc_f}}
\label{sec:proof_inf_regret}
We prove the result for the regret bounds of violation penalty $V_T$ in Theorem~\ref{thm:inf_regret_stc_f} in this section.


We then define the set of phases
\begin{align}
\label{eq:K_s}    
K_s &= \{s \geq 1 \mid a'_s \in \mc A^c, a'_s \in \mc A^\dagger_s \} \\
J_s &= \{s \geq 1 \mid a'_s \in \mc A^c, a'_s \notin \mc A^\dagger_s \} \,,
\end{align}
and based on our construction of algorithm, $J_s$ is a subset of the set of all phases such that $\mc A^\dagger_s = \emptyset$. 

We define the event and the set of phases for $a\in \mc A^0$:
\begin{align}
\label{eq:G_event}
& G_a(t) = \{|\sum_{i}\hat f_{im}(t)a_{im}-\sum_{i}\bar{f}_{im}a_{im}| \leq d^f_{am}(t) \ \forall m  \} \,\\   
\label{eq:Q_event}
& Q_a(t_s)  = \{ a \in \mc A^\dagger_s \}\\
\label{eq:G_phase}
& G_s(a) = \{s \geq 1 \mid G_a(t_s) \text{ holds} \} 
\end{align}

We denote an infeasible assignment $a$ that is expected to return a higher per-unit reward than optimal as a \emph{deceiver assignment}, i.e.,  $\Sigma(q\odot a) > \Sigma(q\odot a^*)$. We then divide the violation penalties as below:
\begin{align}
\nonumber  R^{v}_T = & \E \Big[\sum_{t=1}^T \sum_{m=1}^M  \mid \bar{f}_{im} (a_{im}(t)+b_{im}(t)) - L_m \mid  \Big] \\
    & \leq \sum_{(i,m) \in I_{\mc A_0}}\Delta_{im}^L \cdot B\cdot C_u +R^{dec}_T+ R^{si}_T
\end{align}
where $R^{dec}_T$ stands for regrets incurred from deceiver assignments and $R^{si}_T$ stands for regrets incurred from infeasible suboptimal assignments, and formally,
\begin{align}
\nonumber    R^{dec}_T &= \E\Big[\sum_{s=1}^S \sum_{t=t_s}^{t_{s+1}-1}   \Delta^L_{a'_s} \mathds{1}[K_s] \Big] \\
\nonumber   R^{si}_T &=  \E\Big[\sum_{s=1}^S \sum_{t=t_s}^{t_{s+1}-1}  \Delta^L_{a'_s} \mathds{1}[J_s]  \Big]
\end{align}
where $S$ is the phase that includes round $T$, $t_{S+1} = T+1$. Similar to the definition of $T_{im}^f(t)$, we define $T_{a}^f(t)$ as the number of rounds that the task assignment $a$ has been executed.

Following the proof in~\cite[Lemma 4.3]{ito2024bandit} for the estimation of rewards, we can derive that
\begin{align}
\label{eq:prf_f_prob}
    &\Pr[\hat f_{im}(t)-\bar{f}_{im} \geq d^f_{im}(t)]  \leq \frac{1}{t^2}\\
\label{eq:prf_f_prob2}
    &\Pr[|\hat f_{im}(t)-\bar{f}_{im}| \geq d^f_{im}(t)]  \leq \frac{2}{t^2}\,,
\end{align}
and that 
\begin{equation}
\label{eq:prf_4}
   \sum_{s=1}^S \frac{l_s}{t_s^2} \leq \frac{5}{t_1}  \,.
\end{equation}

\begin{lemma}
\label{lem:deceiver_bound}
    $R^{dec}_T \leq \frac{6\ln(T+1) \bar{L}^2}{(\max_{m}\Delta^L_{am})^2}  + 10\bar{L} \max_{a\in \mc A^c}\Delta^L_{a} \E[\frac{1}{t_1}]$. 
\end{lemma}
\begin{proof}
We assume that $a \in \mc A^c$ and $a \in \mc A^\dagger_s$ for some $s=1,2,\cdots, S+1$.


When the event $G_a(t_s)$ (defined in~\eqref{eq:G_event}) holds for some $s$, for all $m$, $\sum_{i}\hat f_{im}(t_s)a_{im} \geq \sum_{i}\bar{f}_{im}a_{im} - d^f_{am}(t_s)$. Suppose this lower bound $\sum_{i}\bar{f}_{im}a_{im} - d^f_{am}(t_s) > L_m + d^f_{am}(t)$ for some $m$, then this leads to $\sum_{i}\hat f_{im}(t_s)a_{im} >L_m + d^f_{am}(t_s)$. Then $a \notin \mc A^\dagger_s$, which contradicts to our assumption. Hence we must have $\sum_{i}\bar{f}_{im}a_{im} - d^f_{am}(t_s) \leq L_m + d^f_{am}(t)$ for all $m$, and thus $2 d^f_{am}(t_s) \geq \sum_{i}\bar{f}_{im}a_{im} - L_m$. Further, since $a \in \mc A^c$, there exists at least $m' \in I_\mc M$ such that $ \Delta^L_{am} >0$, and thus $2 d^f_{am'}(t_s) \geq \sum_{i}\bar{f}_{im'}a_{im'} - L_{m'} = \Delta^L_{am'} >0$. We then have
$\sqrt{\frac{1.5\bar{L}^2 \ln t_s }{\min_{i:a_{im}=1}{T^f_{im}(t_s)}}} \geq \frac{1}{2} \Delta^L_{am'}$, and hence $\min_{i:a_{im}=1}{T^f_{im}(t_s)} \leq 6\ln(t_s) \bar{L}^2/(\Delta^L_{am})^2$. Notice that $T^f_a(t) \leq T^f_{im}(t)$ for all $t$ and all $(i,m)$ such that $a_{im}=1$, we then have 
\begin{equation}
    \label{eq:deceiver_bound1}
    T^f_a(T) \leq \frac{6\ln(T+1) \bar{L}^2}{(\max_{m}\Delta^L_{am})^2} \,.
\end{equation}


We then show $Pr\{G_a^c(t)\}\leq  \frac{2\bar{L}}{t_s^2}$ in the following.

\begin{align}
\nonumber    G_a^c(t) &=\{\exists m, |\sum_{i}\hat f_{im}(t)a_{im}-\sum_{i}\bar{f}_{im}a_{im}| > d^f_{am}(t)  \} \\
\nonumber   & \subseteq \bigcup_m \{ \exists i, |\hat{f}_{im}(t_s)  - \bar{f}_{im}(t_s)| a_{im}  > \frac{1}{\bar{L}}d_{am}^f(t_s)\} \\
\label{eq:proof_f_a2i}   & \subseteq \bigcup_m \{ \exists i, |\hat{f}_{im}(t_s)  - \bar{f}_{im}(t_s) | a_{im}  > d_{im}^f(t_s)\} \\
\label{eq:prf_inf_1}    & \subseteq \bigcup_{(i,m): a_{im}=1} \{ |\hat{f}_{im}(t_s)  - \bar{f}_{im}(t_s) | >  d_{im}^f(t_s) \},
\end{align}
where~\eqref{eq:proof_f_a2i} follows from that $ d^f_{am}(t)/\bar{L}  =   \sqrt{\frac{1.5\ln t }{\min_{i:a_{im}=1}{T^f_{im}(t)}}}  \geq \sqrt{\frac{1.5 \ln t }{T^f_{im}(t)}} = d^f_{im}(t)$ for all $i$ such that $a_{im}=1$.

Given~\eqref{eq:prf_inf_1} and~\eqref{eq:prf_f_prob2}, we then have
\begin{align}
\nonumber    Pr\{G^c_{a}&(t_s)\} \\
\nonumber   & \leq  Pr\{ \bigcup_{(i,m): a_{im}=1} \{ |\hat{f}_{im}(t_s)  - \bar{f}_{im}(t_s) | >  d_{im}^f(t_s) \} \} \\
\nonumber   & \leq \sum_{(i,m): a_{im}=1} Pr\{  \{| \hat{f}_{im}(t_s)  - \bar{f}_{im}(t_s) | >  d_{im}^f(t_s) \} \} \\
\label{eq:prf_subopt_probGc}   & \leq \sum_{(i,m): a_{im}=1} \frac{2}{t_s^2} \leq  \frac{2\bar L}{t_s^2}\,.
\end{align}

Therefore,
\begin{align}
\nonumber    R^{dec}_T &= \E\Big[\sum_{s=1}^S \sum_{t=t_s}^{t_{s+1}-1}   \Delta^L_{a'_s} \mathds{1}[K_s] \Big] \\
\nonumber    & = \E\Big[\sum_{s=1}^S  l_s \Delta^L_{a'_s} \big( \mathds{1} [K_s \cap G_s(a'_s)]  + \mathds{1}[K_s \cap G^c_s(a'_s)] \big) \Big] \\
\label{eq:prf_inf_3}    & \leq \sum_{a\in \mc A^c} \Delta^L_a \frac{6\ln(T+1) \bar{L}^2}{(\max_{m}\Delta^L_{am})^2}  + \E\Big[\sum_{s=1}^S  l_s \Delta^L_{a'_s}\frac{2\bar{L}}{t_s^2}\Big] \\
\nonumber    & \leq \sum_{a\in \mc A^c}\frac{6 \Delta^L_a \ln(T+1) \bar{L}^2}{(\max_{m}\Delta^L_{am})^2}  + 2\bar{L} \max_{a\in \mc A^c}\Delta^L_{a} \E\Big[\sum_{s=1}^S \frac{l_s}{t_s^2}\Big] \\
\label{eq:prf_inf_2}    & \leq \sum_{a\in \mc A^c} \frac{6  \Delta^L_a \ln(T+1) \bar{L}^2}{(\max_{m}\Delta^L_{am})^2}  + 10\bar{L} \max_{a\in \mc A^c}\Delta^L_{a} \E[\frac{1}{t_1}]
\end{align}
where we derive~\eqref{eq:prf_inf_3} from~\eqref{eq:deceiver_bound1} and~\eqref{eq:G_phase} and that $Pr\{G_a^c(t)\}\leq  \frac{2\bar{L}}{t_s^2}$ for all $a\in \mc A^c\cap \mc A^\dagger_s$, and we get~\eqref{eq:prf_inf_2} from~\eqref{eq:prf_4}.
\end{proof}

\begin{lemma}
\label{lem:R_inf_sub_bound}
    $R^{si} \leq 5\bar{L} \max_{a\in \mc A^c}\Delta^L_{a} \E[\frac{1}{t_1}]$
\end{lemma}
\begin{proof}
We first show that the event $Q^c_{a}(t_s)$ defined in~\eqref{eq:Q_event} has probability $Prob\{Q^c_{a}(t_s)\} \leq  \frac{\bar{L}}{t_s^2}\leq  \frac{NM}{t_s^2}$ for all $a \in \mc A$ in the following:
\begin{align}
\nonumber   Q^c_{a}(t_s) &= \{ a \notin \mc A^\dagger_s \}\\
\nonumber   &= \{\exists m,  \sum_{\forall i} \hat{f}_{im}(t_s) a_{im} - d_{am}^f(t_s) > L_m\}\\
\label{eq:proof_L2f}   &\subseteq \{\exists m,  \sum_{\forall i} (\hat{f}_{im}(t_s)  - \bar{f}_{im}(t_s) )a_{im}  > d_{am}^f(t_s)\}\\
\label{eq:proof_f_a2im}   & = \bigcup_{(i,m): a_{im}=1} \{ \hat{f}_{im}(t_s)  - \bar{f}_{im}(t_s)  >  d_{im}^f(t_s) \} 
\end{align}
where~\eqref{eq:proof_L2f} follows from that when $a \in \mc A$, $\sum_{\forall i} \bar{f}_{im}(t_s) a_{im} - d_{am}^f(t_s) \leq L_m$ for all $m$, 
and all other derivations are the same as those in~\eqref{eq:prf_inf_1}.

Given~\eqref{eq:proof_f_a2im} and \eqref{eq:prf_f_prob}, we then have 
\begin{equation}
\label{eq:prf_subopt_probQc}
 Pr\{Q^c_{a}(t_s)\} 
 \leq \sum_{(i,m): a_{im}=1} \frac{1}{t_s^2} 
 \leq  \frac{\bar L}{t_s^2} 
 <  \frac{NM}{t_s^2} \,.
\end{equation}

For any $s>1$, when $J_s$ holds, then $a'_s \in \mc A^c, a'_s \notin \mc A^\dagger_s$, and thus $\mc A^\dagger_s = \emptyset$. Since we know that $\mc A \neq \emptyset$, then for all $a \in \mc A$, we have $a \notin \mc A^\dagger_s$. 
As a result,
\begin{align}
\nonumber &Pr\{\mc A^\dagger_s = \emptyset \} \leq Pr\{\bigcap_{a\in \mc A} \{a \notin \mc A^\dagger_s\} \}\\
\nonumber & \qquad \leq Pr\{\exists a \in \mc A,  a \notin \mc A^\dagger_s \} \leq Pr\{Q^c_{a}(t_s)\} \leq \frac{\bar L}{t_s^2}
\end{align}
where we derive the last line from~\eqref{eq:proof_f_a2im} and~\eqref{eq:prf_subopt_probQc}.

Hence, similar to the derivation of~\eqref{eq:prf_inf_2}, $R^{si}= \E\Big[\sum_{s=1}^S l_s \Delta^L_{a'_s} \mathds{1}[J_s]  \Big]  \leq 5\bar{L} \max_{a\in \mc A^c}\Delta^L_{a} \E[\frac{1}{t_1}]$.
\end{proof}

Given Lemma~\ref{lem:deceiver_bound} and~\ref{lem:R_inf_sub_bound}, we can then derive the violation regret upper bound as below.
\begin{align}
\nonumber  V_T & = \sum_{(i,m) \in I_{\mc A_0}}\Delta_{im}^L \cdot C_u B +R^{dec}_T+ R^{si}_T\\
\nonumber  & = \sum_{(i,m) \in I_{\mc A_0}}\Delta_{im}^L \cdot C_u  B + \frac{6\ln(T+1) \bar{L}^2}{(\max_{m}\Delta^L_{am})^2} \\
\nonumber & \quad + 10\bar{L} \max_{a\in \mc A^c}\Delta^L_{a} \E[\frac{1}{t_1}] + 5\bar{L} \max_{a\in \mc A^c}\Delta^L_{a} \E[\frac{1}{t_1}]\\
\nonumber & = \mc O\Big( \frac{C_u^2}{C_l}\ln(T) + \frac{\ln(T) \bar{L}^2}{(\max_{m}\Delta^L_{am})^2} + \bar{L} \max_{a\in \mc A^c}\Delta^L_{a} \Big)
\end{align}

\section{Proof for Theorem~\ref{thm:sub_optimal_bound}}
\label{sec:proof_subopt_regret_bound}

We first denote the optimal UCB task assignment at round $t_s$ as $\hat a^*_{s}$, i.e.,
\begin{equation}
\label{eq:UCB_best_sol}     
\hat a^*_{s} \in \argmax_{a \in \mc A_s^\dagger} \Sigma(\hat{q}(t_s) \odot a) 
\end{equation} 

Notice that, based on our assumption that Algorithm A yields a $(1+\alpha)$-approximate solution to~\eqref{eq:UCB_opt}, we have 
\begin{equation}
\label{eq:approx_bound}
    \frac{1}{1+\alpha} \sum_{\hat a^*_{s,im}=1} \hat q_{im} (t_s)\leq\sum_{a'_{s,im}=1}\hat q_{im} (t_s)\,.
\end{equation}


We further define the index set for the optimal and chosen assignments  
\begin{equation}
\label{eq:ind_chosen}     
\begin{split}
A'_s = \{(i,m) \in I_{\mc A_0} \mid a'_{s,im} =1\} \,, \\
A^* = \{(i,m) \in I_{\mc A_0} \mid a^*_{im} =1\} \,.
\end{split}
\end{equation} 
We further let $\tilde{A}^0_s = A'_s \backslash A^*_s$ when $\alpha =0$ and $\tilde{A}^\alpha_s = A'_s$ when $\alpha >0$.

For simplicity of expression, we define the quantity 
\begin{equation}
\label{eq:d}     
d_{im}(t) = \sqrt{\frac{\tilde{C}_{im}}{\ol{c}_{im}} \frac{\ln t}{T_{im}(t)}} \,,
 \end{equation}
where
\begin{equation}
\label{eq:c_tilde}
\begin{split}
    \tilde{C}_{im} &= \frac{1}{\ol{c}_{im}} \Big(2\sqrt{1.5} + \frac{4}{\ol{c}_{im}} \Big(\sqrt{3\sigma_{im}^2} \\
    &\qquad \qquad +15(C_u-C_l) \sqrt{\frac{C_l}{90 C_u}} \Big)  \Big)^2 \,,      
\end{split}
\end{equation}
Notice that
\begin{equation}
\label{eq:c_tilde_bound}
\begin{split}
    \tilde{C}_{im} &\leq \mc O\Big(\frac{1}{\ol{c}_{im}} \Big( 1+ \frac{1}{\ol{c}^2_{im}} \Big(\sigma^2_{im} + \frac{(C_u-C_l)^2 C_l}{C_u} \Big) \Big)\Big) \\
    & \leq \mc O\Big(\frac{1}{\ol{c}_{im}} \Big( 1+  \frac{C_u-C_l}{\ol{c}_{im}} \Big)\Big) \leq \mc O\Big(  \frac{C_u}{\ol{c}_{im}C_l} \Big). \,      
\end{split}
\end{equation}
We further define two designated sets of phases 
\begin{equation}
\label{eq:F_s}    
\mc F_s = \{s \geq 1 \mid \Delta^\alpha_{a'_s} \leq \sum_{(i,m) \in \tilde{A}^\alpha_s} d_{im}(t_s) , \Delta^\alpha_{a'_s} >0 \} \, ,
\end{equation}
\begin{equation}
\label{eq:H_s}    
\mc H_s = \{s \geq 1 \mid a^* \in \mc A^\dagger_s\} \cap \mc F_s  \, ,
\end{equation}
and let $\mc H^c_s$ be its complement set.
Finally, we define the reward our algorithm can obtain from the phases in~\eqref{eq:F_s} as $\hat{R}_T $
\begin{equation}
\label{eq:R_hat} 
\hat{R}_T = \sum_{s=1}^S l_s \Delta^\alpha_{a'_s} \1[\mc F_s] \,. 
\end{equation}

\begin{lemma}
\label{lem:subopt_R1_bound}
    For all phases $s \geq 1$, we have $\Delta_{a'_s} \leq \sum_{(i,m) \in \tilde{A}_s} d_i(t_s)$ with a probability at least $1-\frac{7MN}{t_s^2}$. We then have $R^{(1)}_T \leq \E[\hat{R}_T] +\mc O \big(\frac{MN\ol L}{(1+\alpha)C_l t_1} \big)$
\end{lemma}

\begin{proof}
We first show $\Delta_{a'_s} \leq \sum_{(i,m) \in \tilde{A}_s} d_{im}(t_s)$ with a probability $ 1- \frac{7MN}{t_s^2}$ for all $s \geq 1$.

Following the proof in~\cite[Lemma 4.3]{ito2024bandit}, we have~\eqref{eq:q_upper} below. 
\begin{equation}
\label{eq:q_upper}
    \hat q_{im}(t_s) \leq q_{im}(t_s) + d_{im}(t_s) \,.
\end{equation}

We first note that for all $\alpha \geq 0$,
\begin{equation}
    \label{eq:prf_7}
    \Delta^\alpha_{a'_s} \leq \frac{1}{1+\alpha} \sum_{a^*_{im}=1} q_{im} \leq \frac{\ol L}{(1+\alpha)C_l}
\end{equation}
as $q_{im} \leq 1/C_l$, and there exists at most $\ol L$ pairs of $(i,m)$ such that $a^*_{im}=1$.

We then show that $\Delta_{a'_s} \leq \sum_{(i,m) \in \tilde{A}^\alpha_s} d_{im}(t_s)$ with a probability $ 1- \frac{7MN}{t_s^2}$ in two cases, $\alpha =0$ and $\alpha >0$, separately in the following.

When $\alpha =0$, 
\begin{align}
\nonumber  \Delta^0_{a'_s} &  = \sum_{a^*_{im}=1} q_{im} - \sum_{a'_{im}=1} q_{im} \\
\label{eq:prf_1-1}   &  = \sum_{(i,m) \in A^* \backslash A'_s} q_{im} - \sum_{(i,m) \in A'_s \backslash A^*} q_{im} \\
\label{eq:prf_1-2}      & \leq  \sum_{(i,m) \in A^* \backslash A'_s} \hat q_{im}(t_s) -\sum_{(i,m) \in A'_s \backslash A^*} q_{im}
\end{align}

When $a^* \in \mc A^\dagger_s$, $\mc A^\dagger_s \neq \emptyset$. According to our algorithm, $a'_s \in \mc A^\dagger_s$ and $a'_s \in \argmax_{a \in \mc A^\dagger_s} \Sigma(\hat{q}(t_s) \odot a)$, as a result, $\Sigma(\hat{q}(t_s) \odot a'_s) \geq \Sigma(\hat{q}(t_s) \odot a^*)$. And also, $Pr\{ Q^c_{a^*}(t_s) \} \leq  \frac{NM}{t_s^2}$
Therefore, from~\eqref{eq:prf_1-2}, with probability at least $1- \frac{7NM}{t_s^2}$, we have
\begin{align}
\nonumber  \Delta_{a'_s} & \leq  \sum_{(i,m) \in A^* \backslash A'_s} \hat q_{im}(t_s) -\sum_{(i,m) \in A'_s \backslash A^*} q_{im}\\
\nonumber & \leq  \sum_{(i,m) \in  A'_s\backslash A^*} \hat q_{im}(t_s) -\sum_{(i,m) \in A'_s \backslash A^*} q_{im} \\
\label{eq:proof_1-4} & = \sum_{(i,m) \in \tilde{A}_s} d_{im}(t_s) \,.
\end{align}

When $\alpha>0$, similar to the case when $\alpha =0$, since $a^* \in \mc A^\dagger_s$ with probability greater then $1-\frac{NM}{t_s^2}$, and $a^* = \argmax_{a\in \mc A^\dagger_s} \Sigma(\hat{q} \odot a)$ under this event, then with probability at least $1- \frac{7NM}{t_s^2}$, we have
\begin{align}
\nonumber  \Delta^\alpha_{a'_s}&  = \frac{1}{1+\alpha} \sum_{a^*_{s,im}=1} q_{im} - \sum_{a'_{s,im}=1} q_{im} \\
\label{eq:prf_1}      & \leq \frac{1}{1+\alpha} \sum_{a^*_{s,im}=1} \hat q_{im}(t_s) -\sum_{a'_{s,im}=1} q_{im}\\
\label{eq:prf_2}      & \leq  \frac{1}{1+\alpha} \sum_{\hat a^*_{s,im}=1} \hat q_{im}(t_s) -\sum_{a'_{s,im}=1} q_{im} \\
\nonumber   & =  \frac{1}{1+\alpha} \sum_{\hat a^*_{s,im}=1} \hat q_{im} (t_s)-\sum_{a'_{s,im}=1}\hat q_{im}(t_s)\\
\nonumber   & \qquad \qquad  +\sum_{a'_{s,im}=1}(\hat q_{im}(t_s) - q_{im}(t_s) ) \\     
\label{eq:prf_3}   & \leq \sum_{a'_{s,im}=1}\hat q_{im}(t_s)-\sum_{a'_{s,im}=1}\hat q_{im}(t_s) + \sum_{a'_{s,im}=1} d_{im}(t_s)
\end{align}
where~\eqref{eq:prf_1} is true with probability $1- \frac{6NM}{t^2}$ based on~\eqref{eq:berns_ineq}. \eqref{eq:prf_2} holds based on the optimum definition of $\hat a^*_{im}=1$ in~\eqref{eq:UCB_best_sol} when $a^* \in \mc A^\dagger_s$, and~\eqref{eq:prf_3} follows from~\eqref{eq:approx_bound}.

Next, we rewrite and bound $R^{(1)}_T$ in the rest of the proof.
\begin{align}
\nonumber R^{(1)}_T &= \E\Big[\sum_{s=1}^S \sum_{t=t_s}^{t_{s+1}-1}  \Sigma \Big( \frac{1}{\alpha+1} q \odot a^* - q\odot  a_{s}' \Big) \Big] \\
\nonumber & = \E\Big[\sum_{s=1}^S  l_s   \Delta_{a'_s}  \Big] \\
\nonumber& = \E\Big[\sum_{s=1}^S  l_s   \Delta^\alpha_{a'_s} \1[\mc H_s] + \sum_{s=1}^S  l_s   \Delta_{a'_s} \1[\mc H^c_s] \Big] \\
\label{eq:prf_R_1_bound} & = \E[\hat R_T] + \E\Big[\sum_{s=1}^S  l_s   \Delta_{a'_s} \1[\mc H^c_s] \Big] \,,
\end{align}
where $\mc H_s$ is defined in~\eqref{eq:H_s} and
\begin{align}
\nonumber & \E\Big[\sum_{s=1}^S  l_s \Delta^\alpha_{a'_s} \1[\ol {\mc H}_s] \Big] \\
\nonumber  & \leq \E\Big[\sum_{s=1}^S  l_s \Delta^\alpha_{a'_s} \1[\Delta_{a'_s} > \sum_{(i,m) \in \tilde{A}^\alpha_s} d_{im}(t_s) ] \Big]\\ 
\label{eq:prf_sub_5} & \leq \E\Big[\sum_{s=1}^S  l_s \frac{\ol L}{(1+\alpha)C_l} \frac{7MN}{t_s^2} \Big] 
= \frac{7MN\ol L}{(1+\alpha)C_l} \E\Big[\sum_{s=1}^S \frac{l_s}{t_s^2}  \Big] \\
\label{eq:prf_sub_6} & \leq \frac{35MN\ol L}{(1+\alpha)C_l} \E\Big[ \frac{1}{t_1} \Big] \,, 
\end{align}
where~\eqref{eq:prf_sub_5} follows from that~\eqref{eq:proof_1-4} and~\eqref{eq:prf_3} holds with probability at least $1- \frac{7NM}{t_s^2}$ and~\eqref{eq:prf_7}, and~\eqref{eq:prf_sub_6} follows from~\eqref{eq:prf_4}\,.

Then~\eqref{eq:prf_R_1_bound} and~\eqref{eq:prf_sub_6} together yields the bound on $R^{(1)}_T$ in Lemma~\ref{lem:subopt_R1_bound}. 
\end{proof}

To analyze the upper bound of $\E[\hat{R}_T]$, we first obtain Lemma~\ref{lem:count_lowerbound} by directly applying the proof for~\cite[Lemma A.4]{ito2024bandit} to the lemma by revising the task $i$ in its proof to agent-task pair $(i,m) \in I_\mc A$.
\begin{lemma}
\label{lem:count_lowerbound}
For all $s \leq 1$ and $(i,m) \in  A'_s $, we have 
\begin{equation}
    \label{eq:count_lower}
    \E[T_{im}(t_{s+1})- T_{im}(t_{s}) \mid l_s] \geq \frac{l_s - 2C_u}{\ol{c}_{im}} \,.
\end{equation}
Therefore, we have
\begin{equation}
    \label{eq:ls_upper}
    l_s \leq 2\ol{c}_{im} \E[T_{im}(t_{s+1})-T_{im}(t_{s}) \mid l_s,T_{im}(t_{s}) ]\,.
\end{equation}
\end{lemma}

We now bound $\E[\hat{R}_T]$ used in the upper bound of $R^{(1)}_T$ in Lemma~\ref{lem:subopt_R1_bound}.
\begin{lemma}
\label{lem:subopt_Rhat_bound}
    $\E[\hat{R}_T] = \mc O\Big(\ol{L} \sum_{(i,m) \in \tilde{A}^\alpha_s} \frac{\tilde{C}_{im} \ln T}{\Delta^\alpha_{im}} \Big)$
\end{lemma}
\begin{proof}
By leveraging~\cite[Lemma A.2]{ito2024bandit} with $\{\alpha_k\}$ and $\{\beta_k\}$ such that $\alpha_1>\alpha_2 >\hdots >0$, $1=\beta_0 > \beta_1>\beta_2 >\hdots >0$, and $\sum_{k=1}^\infty \frac{\alpha_k}{\beta_k} = \mc O(1)$, if $\mc F_s$ defined in~\eqref{eq:F_s} occurs, we then have $\sum_{(i,m) \in  \tilde{A}^\alpha_s} \1[(i,m),\Delta^\alpha_{a'_s} \leq \ol{L} \sqrt{\alpha_k}d_{im}(t_s)] \geq \beta_k$ for some $k$. Additionally, since $\mc H_s \subseteq \mc F_s$
\begin{equation}
    \1[\mc H_s] \leq \1[\mc F_s] \leq \sum_{k=1}^{\infty}\frac{1}{\ol{L}\beta_k} \sum_{(i,m) \in  \tilde{A}^\alpha_s} \1[\mc G_{s,k,im}]
\end{equation}
where $\mc G_{s,k,im} = \{(i,m) \in  \tilde{A}^\alpha_s \mid  \Delta_{a'_s} \leq \sqrt{\alpha_k}d_{im}(t_s) \}$.

We then bound $\E[\hat{R}_T]$ as below
\begin{align}
\nonumber    & \E[\hat{R}_T]  = \E\Big[\sum_{s=1}^S  l_s \Delta^\alpha_{a'_s} \1[\mc H_s]  \Big] \\
\nonumber    & \leq \sum_{(i,m) \in  \tilde{A}^\alpha_s} \sum_{k=1}^{\infty}\frac{1}{\ol{L}\beta_k} \sum_{s=1}^S l_s\Delta^\alpha_{a'_s} \1[\mc G_{s,k,im}] \\
\nonumber    & \leq \sum_{(i,m) \in  \tilde{A}^\alpha_s} \sum_{k=1}^{\infty}\frac{1}{\beta_k} \sum_{s=1}^S l_s \sqrt{\alpha_k}d_{im}(t_s) \1[\mc G_{s,k,im}] \\
\nonumber    & = \sum_{(i,m) \in  \tilde{A}^\alpha_s} \sum_{k=1}^{\infty}\frac{\sqrt{\alpha_k}}{\beta_k} \sum_{s=1}^S l_s \sqrt{\frac{\tilde{C}_{im}}{\ol{c}_{im}} \frac{\ln t_s}{T_{im}(t_s)}}  \1[\mc G_{s,k,im}] \\
\label{eq:R_hat_pf_d}    & \leq \sqrt{\ln T} \sum_{(i,m) \in \tilde{A}^\alpha_s}\sqrt{\frac{\tilde{C}_{im}}{\ol{c}_{im}}} \sum_{k=1}^{\infty}\frac{\sqrt{\alpha_k}}{\beta_k} \sum_{s=1}^S \frac{l_s  \1[\mc G_{s,k,im}] }{\sqrt{T_{im}(t_s)}} \,.
\end{align}

With given $k$ and $i$, we then bound the term $\sum_{s=1}^S \frac{l_s  \1[\mc G_{s,k,im}] }{\sqrt{T_{im}(t)}}$ in~\eqref{eq:R_hat_pf_d}. Suppose that $\mc G_{s,k,im}$ occurs, we then have $\Delta^\alpha_{im} \leq \Delta_{a'_s} \leq \ol{L} \sqrt{\alpha_k}d_{im}(t_s) \leq \ol{L} \sqrt{\frac{\alpha_k \tilde{C}_{im}\ln T }{\ol{c}_{im}T_{im}(t_s)} }$, where the first inequality follows from the definition of $\Delta_{im}$. As a result, $\sqrt{T_{im}(t_s)} \leq \frac{\ol L}{\Delta_{im}} \sqrt{\frac{\alpha_k \tilde{C}_{im}\ln T }{\ol{c}_{im}}}$, and we define it as $\gamma$. Based on this a sufficient but not necessary condition, we can derive that
\begin{equation}
    \label{eq:R_hat_pf_a}
    \sum_{s=1}^S \frac{l_s  \1[\mc G_{s,k,im}]}{\sqrt{T_{im}(t_s)}} \leq  \sum_{s=1}^S \frac{l_s  \1[ (i,m) \in  \tilde{A}^\alpha_s] \1[\sqrt{T_{im}(t_s)} \leq \gamma]}{\sqrt{T_{im}(t_s)}}\,.
\end{equation}

Using~\eqref{eq:ls_upper} in Lemma~\ref{lem:count_lowerbound}, we then have 

\begin{equation}
\label{eq:R_hat_pf_b}
\begin{split}    
&\E \Big[ \frac{l_s  \1[ (i,m) \in \tilde{A}^\alpha_s]}{\sqrt{T_{im}(t_s)}} \Big] \\
& \leq \E \Big[ \frac{2\ol{c}_{im} \E[T_{im}(t_{s+1})-T_{im}(t_{s}) \mid l_s,T_{im}(t_{s}) ]}{\sqrt{T_{im}(t_s)}} \Big] \\
& = 2\ol{c}_{im} \sum_{\substack{\forall \ l_s, T_{im}(t_{s})}} Pr\{l_s, T_{im}(t_{s})\} \frac{1}{\sqrt{T_{im}(t_s)}}  \cdot \\ &  \sum_{T_{im}(t_{s+1})} Pr\{T_{im}(t_{s+1}) \mid l_s, T_{im}(t_{s})\} (T_{im}(t_{s+1})-T_{im}(t_{s})) \\
& = 2\ol{c}_{im} \sum_{\substack{ l_s, T_{im}(t_{s}),\\ T_{im}(t_{s+1})}} Pr\{l_s, T_{im}(t_{s}), T_{im}(t_{s+1})\}  \cdot \\& \hspace{5cm} \frac{T_{im}(t_{s+1})-T_{im}(t_{s})}{\sqrt{T_{im}(t_s)}} \\
& = 2\ol{c}_{im} \E \Big[ \frac{T_{im}(t_{s+1})-T_{im}(t_{s}) }{\sqrt{T_{im}(t_s)}} \Big]\,,
\end{split}  
\end{equation}
where we have 
\begin{equation*}
\label{eq:R_hat_pf_c}
\begin{split}    
& \frac{T_{im}(t_{s+1})-T_{im}(t_{s}) }{\sqrt{T_{im}(t_s)}}  = 3\frac{T_{im}(t_{s+1})-T_{im}(t_{s}) }{\sqrt{4T_{im}(t_s)}+\sqrt{T_{im}(t_s)}} \\
& \leq 3\frac{T_{im}(t_{s+1})-T_{im}(t_{s}) }{\sqrt{T_{im}(t_{s+1})}+\sqrt{T_{im}(t_s)}} 
 = 3\sqrt{T_{im}(t_{s+1})-T_{im}(t_{s})} \,,
\end{split}  
\end{equation*}
and the inequality follows from Lemma~\ref{lem:phase_length}.

Given~\eqref{eq:R_hat_pf_a}--\eqref{eq:R_hat_pf_c} and define $s_0$ as the phase such that $T_{im}(t_{s_0+1}) > \gamma $ while $T_{im}(t_{s_0}) \leq \gamma $, we then have  
\begin{align*}
&  \sum_{s=1}^S \frac{l_s  \1[\mc G_{s,k,im}]}{\sqrt{T_{im}(t_s)}} \\
&  \leq  2\ol{c}_{im}  \E \Big[ \sum_{s=1}^S \frac{T_{im}(t_{s+1})-T_{im}(t_{s}) }{\sqrt{T_{im}(t_s)}}  \1[\sqrt{T_{im}(t_s)} \leq \gamma]\Big] \\
& \leq 6\ol{c}_{im}  \E \Big[ \sum_{s=1}^S \sqrt{T_{im}(t_{s+1})-T_{im}(t_{s})}   \1[\sqrt{T_{im}(t_s)} \leq \gamma]\Big] \\
& = 6\ol{c}_{im}  \E \Big[ \sum_{s=1}^{s_0}\sqrt{T_{im}(t_{s+1})-T_{im}(t_{s})}   \1[\sqrt{T_{im}(t_s)} \leq \gamma]\Big] \\
& =  6\ol{c}_{im} T_{im}(t_{s_0+1}) 
\leq 12 \ol{c}_{im} T_{im}(t_{s_0+1}) 
\leq 12 \ol{c}_{im}\gamma \\
& =  \frac{12 \ol L \ol{c}_{im}}{\Delta^\alpha_{im}} \sqrt{\frac{\alpha_k \tilde{C}_{im}\ln T }{\ol{c}_{im}}} 
= \frac{12 \ol L }{\Delta^\alpha_{im}} \sqrt{\alpha_k \ol{c}_{im} \tilde{C}_{im}\ln T }\,.
\end{align*}

Then, from~\eqref{eq:R_hat_pf_d} and the above, we conclude that 
\begin{align*}
& \E[\hat{R}_T]  \leq \sqrt{\ln T} \sum_{(i,m) \in \tilde{A}^\alpha_s}\sqrt{\frac{\tilde{C}_{im}}{\ol{c}_{im}}} \sum_{k=1}^{\infty}\frac{\sqrt{\alpha_k}}{\beta_k} \sum_{s=1}^S \frac{l_s  \1[\mc G_{s,k,im}] }{\sqrt{T_{im}(t_s)}} \\
& \leq \sqrt{\ln T} \sum_{(i,m) \in \tilde{A}^\alpha_s}\sqrt{\frac{\tilde{C}_{im}}{\ol{c}_{im}}} \sum_{k=1}^{\infty}\frac{\sqrt{\alpha_k}}{\beta_k} \frac{12 \ol L }{\Delta^\alpha_{im}} \sqrt{\alpha_k \ol{c}_{im} \tilde{C}_{im}\ln T } \\
& = \mc O\Big(\ol L \ln T \sum_{(i,m) \in \tilde{A}^\alpha_s} \frac{\tilde{C}_{im}}{\Delta^\alpha_{im}}  \frac{\alpha_k}{\beta_k}  \Big)
= \mc O\Big(\ol{L} \sum_{(i,m) \in \tilde{A}^\alpha_s} \frac{\tilde{C}_{im} \ln T}{\Delta^\alpha_{im}} \Big) .
\end{align*}
\end{proof}


\begin{lemma}
    \label{lem:R2_bound}
    $R^{(2)}_T = \E\Big[\sum_{s=1}^S \sum_{t=t_s}^{t_{s+1}-1}  \Sigma \Big(q\odot  a_{s}'- r(t)\odot  a(t) \Big)  \mathds{1}\{a(t)+b(t)\in \mc A \} \Big] = \mc O\Big(\frac{C_u^2}{C_l^2}NM\ol{L} \ln T \Big)$
\end{lemma}
\begin{proof}
First, following the proof of~\cite[Lemma A.3]{ito2024bandit}, we can derive that 
\begin{equation} 
\label{eq:R2_bound_pf_0}
    \E\Big[\sum_{t=t_s}^{t_{s+1}-1}  \Sigma \Big(q\odot  a_{s}'- r(t)\odot  a(t) \Big) \Big] \leq 3 \ol L \frac{C_u}{C_l} \,,
\end{equation}
given any $a_{s}', t_s$, and $l_s$. For clarity, we restate the proof in our problem setting in the following. 

Given $s \geq 1$, for any task-member pair $(i,m) \in I_\mc A$, the number of times that the member $m$ will start the task $i$ during the phase $s$ is at least $T_{im}(t_{s+1})-T_{im}(t_{s})-1$. Therefore, we have
\begin{align}
\nonumber &\E\Big[\sum_{t=t_s}^{t_{s+1}-1}  \Sigma \big( r(t) \odot  a(t) \big) \Big] \\
\nonumber &\geq \E\Big[\sum_{(i,m) \in A'_s}  \ol r_{im} (T_{im}(t_{s+1})-T_{im}(t_{s})-1) \Big] \\
\nonumber &\geq \sum_{(i,m) \in A'_s} \ol r_{im} \Big[ \sum_{\substack{ l_s, T_{im}(t_{s}),\\ T_{im}(t_{s+1})}} Pr\{l_s, T_{im}(t_{s}), T_{im}(t_{s+1})\} \cdot \\
\nonumber &  \hspace{4cm}(T_{im}(t_{s+1})-T_{im}(t_{s})-1) \Big] \\
\nonumber &\geq \sum_{(i,m) \in A'_s} \ol r_{im} \Big[  \sum_{\substack{ l_s,\\ T_{im}(t_{s})}}  Pr\{l_s\}\cdot \sum_{ T_{im}(t_{s+1})}  \\
\nonumber &  Pr\{T_{im}(t_{s+1} ), T_{im}(t_{s}))\mid l_s \} (T_{im}(t_{s+1})-T_{im}(t_{s})-1) \Big] \\
\nonumber &= \sum_{(i,m) \in A'_s} \ol r_{im}  \E\Big[T_{im}(t_{s+1})-T_{im}(t_{s})-1 \mid l_s \Big]  \\
\label{eq:R2_bound_pf_a} & \geq \sum_{(i,m) \in A'_s} \ol r_{im}  \E\Big[ \frac{l_s - 2C_u}{\ol{c}_{im}} -1  \Big]  \\
\nonumber & = \sum_{(i,m) \in A'_s} \ol r_{im}  \E\Big[ \frac{l_s - 2C_u-\ol{c}_{im}}{\ol{c}_{im}} \Big] \\
\label{eq:R2_bound_pf_b} & \geq \sum_{(i,m) \in A'_s} \ol r_{im}  \E\Big[ \frac{l_s - 3C_u}{\ol{c}_{im}} \Big] 
= \sum_{(i,m) \in A'_s}  \E\Big[ q_{im} (l_s - 3C_u) \Big] \\
\label{eq:R2_bound_pf_c} & \geq \sum_{(i,m) \in A'_s}  \E\Big[ q_{im} l_s \Big] -3 \frac{C_u}{C_l} \ol L \\
\nonumber & =    \E\Big[ l_s  \Sigma(q_{im} \odot a'_s) \Big] -3 \frac{C_u}{C_l} \ol L \,,
\end{align}
where the inequality~\eqref{eq:R2_bound_pf_a} follows from Lemma~\ref{lem:count_lowerbound}, the inequality in~\eqref{eq:R2_bound_pf_b} follows from that $\ol{c}_{im} \leq C_u$,~\eqref{eq:R2_bound_pf_c} follows from that the team can execute at most $\ol L$ tasks at the same time, while $q_{im} \leq 1/C_l$. We then directly obtain~\eqref{eq:R2_bound_pf_0} and derive the upper bound for $R^{(2)}_T$.
\begin{align*}
R^{(2)}_T &= \E\Big[\sum_{s=1}^S \sum_{t=t_s}^{t_{s+1}-1}  \Sigma \Big(q\odot  a_{s}'- r(t)\odot  a(t) \Big)\cdot \\
& \hspace{4cm} \mathds{1}\{a(t)+b(t)\in \mc A \} \Big] \\
& \leq \E\Big[\sum_{s=1}^S \sum_{t=t_s}^{t_{s+1}-1}  \Sigma \Big(q\odot  a_{s}'- r(t)\odot  a(t) \Big)\Big] \\
& \leq S\cdot 3 \frac{C_u}{C_l} \ol L     
\leq  \mc O (\frac{C_u}{C_l}NM\ln T  \cdot 3 \frac{C_u}{C_l} \ol L) \\
& =  \mc O (\frac{C_u^2}{C_l^2}NM \ol L\ln T )\,.
\end{align*}
where the last inequality that bounds the number of total phases $S$ follows from~\eqref{eq:S_bound}.
\end{proof}

\begin{lemma}
    \label{lem:R3_bound}
    $R^{(3)}_T =  \E\Big[\sum_{s=1}^S \sum_{t=t_s}^{t_{s+1}-1}  \Sigma \Big(q\odot  a_{s}'\Big) \mathds{1}\{a(t)+b(t)\in \mc A^c\} \Big] = \mc O \Big(\sum_{a \in \mc A^c} \frac{\bar{L}^3\ln T}{C_l(\max_{m}\Delta^L_{am})^2} +\frac{C^2_u \bar{L}}{C^2_l}NM\ln T \Big)$
\end{lemma}
\begin{proof}

\begin{align}
\nonumber   & \E\Big[\sum_{s=1}^S \sum_{t=t_s}^{t_{s+1}-1}  \Sigma (q\odot  a_{s}') \mathds{1}\{a(t)+b(t)\in \mc A^c\} \Big]\\
\label{eq:R3prf_1} & \leq \frac{C_u S \bar{L} }{C_l} +\E\Big[\sum_{s=1}^S \sum_{t=t_s}^{t_{s+1}-1}  \Sigma (q\odot  a_{s}') \mathds{1}\{a_{s}' \in \mc A^c\} \Big] \\
\nonumber & = \frac{C_u S \bar{L} }{C_l} + \E\Big[\sum_{s=1}^S \sum_{t=t_s}^{t_{s+1}-1}  \Sigma (q\odot  a_{s}') \big[\mathds{1}\{K_s, G_{a_{s}'}(t_s)\} \\
\nonumber & \quad +\mathds{1}\{K_s, G^c_{a_{s}'}(t_s)\} +\mathds{1}\{J_s\} ] \Big] \\
\nonumber & = \E\Big[\sum_{s=1}^S \sum_{t=t_s}^{t_{s+1}-1}  \Sigma (q\odot  a_{s}')\mathds{1}\{K_s, G_{a_{s}'}(t_s)\} \Big]  \\
\label{eq:R3prf_2} & \quad + \frac{C_u S \bar{L} }{C_l} + \E\Big[\sum_{s=1}^S \sum_{t=t_s}^{t_{s+1}-1} \frac{3 \bar L^2}{C_l t_s^2} \Big] \\
\label{eq:R3prf_3} & \leq \sum_{a \in \mc A^c} \E \Big[ \Sigma (q\odot  a) T_a^f(T)\} \Big]  +\frac{3 \bar L^2}{C_l}  \E\Big[\sum_{s=1}^S \frac{l_s}{ t_s^2} \Big] + \frac{C_u S \bar{L} }{C_l}  \\
\label{eq:R3prf_4} & \leq  \sum_{a \in \mc A^c} \frac{\bar L}{C_l } \cdot \frac{6\ln(T+1) \bar{L}^2}{(\max_{m}\Delta^L_{am})^2} + \frac{15 \bar L^2}{C_l}  \E \Big[\frac{1}{t_1}\Big] + \frac{C_u S \bar{L} }{C_l}  \\
\label{eq:R3prf_5} & \leq \mc O \Big(\sum_{a \in \mc A^c} \frac{\bar{L}^3\ln T}{C_l(\max_{m}\Delta^L_{am})^2} +\frac{C^2_u \bar{L}}{C^2_l}NM\ln T \Big) \,,
\end{align}
where the~\eqref{eq:R3prf_1} is because for each phase, $a(t) +b(t) \neq a'_s$ for at most $C_u$ rounds at the beginning, and that $\Sigma (q\odot  a_{s}') \leq \bar{L} / C_l$.
From~\eqref{eq:prf_subopt_probQc}, we have $Pr\{G_a^c(t)\}\leq  \frac{2\bar{L}}{t_s^2}$, and from the proof of Lemma~\ref{lem:R_inf_sub_bound}, we have $ Pr\{J_s\} \leq \frac{\bar L}{t_s^2}$, and we then obtain~\eqref{eq:R3prf_2}. We derive~\eqref{eq:R3prf_3} from~\eqref{eq:prf_4}.
Further, we have $T^f_a(T) \leq \frac{6\ln(T+1) \bar{L}^2}{(\max_{m}\Delta^L_{am})^2}$ from~\eqref{eq:deceiver_bound1}, we then derive~\eqref{eq:R3prf_4}. Next, by~\eqref{eq:S_bound}, we have $S \leq \mc O (\frac{C_u}{C_l}NM\ln T)$, and further, $1/t_1 \leq 1$ and $\bar{L} \leq NM$, we then derive~\eqref{eq:R3prf_5}. We therefore complete the proof.
\end{proof}

Given Lemmas~\ref{lem:subopt_R1_bound}, \ref{lem:subopt_Rhat_bound}, \ref{lem:R2_bound} , and~\ref{lem:R3_bound}, we then cast bounds on $R^\alpha_T$ following $R^{(1)}_T$, $R^{(2)}_T$, and $R^{(3)}_T$. 
We show the proof for Theorem~\ref{thm:sub_optimal_bound} below.
\begin{proof}
When the Algorithm~\eqref{alg:bandit} is an exact algorithm, i.e., $\alpha =0$. For this problem setting, we can consider the multi-agent problem as a generalization of task assignment for a single agent with $NM$ tasks and maximal capacity $\ol{L}$, with extra constraints on the set of feasible actions to ensure that the planner will not assign the same task to multiple agents simultaneously, and additional concern of the stochastic resources that each task takes from each agent. We can then leverage the analysis from~\cite[Theorem 4.5]{ito2024bandit} to obtain the lower bound on the regret $R^0_T$.

We then analyze the upper bound for the regret.
Given Lemmas~\ref{lem:subopt_R1_bound} and~\ref{lem:subopt_Rhat_bound}, we have 
\begin{align}
\nonumber    R^{(1)}_T &= \E\Big[\sum_{s=1}^S \sum_{t=t_s}^{t_{s+1}-1}  \Sigma \Big(\frac{1}{\alpha+1} q \odot a^* - q\odot  a_{s}' \Big) \Big] \\
\nonumber & \leq \E[\hat{R}_T] +\mc O \big(\frac{MN\ol L}{C_l t_1} \big) \\
\nonumber & = \mc O\Big(\ol{L} \sum_{(i,m) \in \tilde{A}^\alpha_s} \frac{\tilde{C}_{im} \ln T}{\Delta^\alpha_{im}} + \frac{MN\ol L}{C_l t_1} \Big) \\
\label{eq:thm_pf_a} & = \mc O\Big(\ol{L} \sum_{(i,m) \in \tilde{A}^\alpha_s} \frac{\tilde{C}_{im} \ln T}{\Delta^\alpha_{im}} + \frac{NM\ol{L}}{C_l} \Big) \,.
\end{align}

Combining this with Lemma~\ref{lem:R2_bound} and Lemma~\ref{lem:R3_bound}, and that $\E[t_1] = \mc O \Big(\frac{C^2_u N M \ln T}{C_l}  \Big)$ from~\eqref{eq:basic_time}, we then get the gap-dependent bound $R^\alpha_T \leq \mc O\Big(\ol{L} \sum_{ \forall m,i } \frac{\tilde{C}_{im} \ln T }{\Delta^\alpha_{im}} +\frac{C_u^2}{C_l^2}NM\ol{L} \ln T +\sum_{a \in \mc A^c} \frac{\bar{L}^3\ln T}{C_l(\max_{m}\Delta^L_{am})^2} \Big) 
\leq \mc O\Big(\frac{\max_{ \forall m,i }\tilde{C}_{im} }{\ul \Delta^\alpha } NM\ol{L} \ln T  +\frac{C_u^2}{C_l^2}NM\ol{L} \ln T  +\sum_{a \in \mc A^c} \frac{\bar{L}^3\ln T}{C_l(\max_{m}\Delta^L_{am})^2}\Big)
\leq \mc O\Big(\Big(\frac{1 }{\ul \Delta^\alpha } + C_u\Big)\frac{C_u}{C_l^2}NM\ol{L} \ln T+\sum_{a \in \mc A^c} \frac{\bar{L}^3\ln T}{C_l(\max_{m}\Delta^L_{am})^2} \Big).$

Lastly, to obtain the gap-independent regret bound, we modify the analysis of $\hat R_T$. Given any $\epsilon >0$, we can then bound $\hat R_T$ from above
\begin{align}
\nonumber \E & [\hat{R}_T] = \E\Big[\sum_{s=1}^S  l_s \Delta_{a'_s} \1[\mc H_s]  \Big] \\
\nonumber & \leq \sum_{s=1}^S  \E\Big[l_s \Delta^\alpha_{a'_s} \1[\mc H_s, \Delta^\alpha_{a'_s} >\epsilon] + l_s \Delta^\alpha_{a'_s} \1[\mc H_s, \Delta^\alpha_{a'_s} \leq \epsilon] \Big] \\
\label{eq:thm_pf_b} & \leq \mc O\Big(\ol{L} \sum_{(i,m) \in \tilde{A}_s} \frac{\tilde{C}_{im} \ln T}{\Delta^\alpha_{im}}\1[ \Delta^\alpha_{a'_s} > \epsilon] \Big] +  \epsilon T\Big)  \\
\nonumber & \leq \mc O\Big(\ol{L} \sum_{(i,m) \in \tilde{A}_s} \frac{\tilde{C}_{im} \ln T}{\epsilon } + \epsilon T \Big) 
\end{align}
where the first part of~\eqref{eq:thm_pf_b} follows exactly from the proof of Lemma~\ref{lem:subopt_Rhat_bound}, and the latter part of~\eqref{eq:thm_pf_b} follows from that $l_s <T$. Then by setting $\epsilon = \sqrt{\frac{C_u MN \ol L\ln T}{C_l^2 T}}$, we obtain that $\E [\hat{R}_T] \leq \mc O \Big( \frac{1}{C_l} \sqrt{C_u MN \ol L T \ln T} \Big) $. And we then get the gap-free bound for $R^{\alpha}_T$ as in Theorem~\ref{thm:sub_optimal_bound} that
\begin{align*}  
        R^{\alpha}_T \leq & \mc O\Big(\frac{1}{C_l} \sqrt{ C_u N M\ol{L} T\ln T} + \frac{C^2_u}{C_l^2}NM\ol{L} \ln T  \\ &\qquad +\sum_{a \in \mc A^c} \frac{\bar{L}^3\ln T}{C_l(\max_{m}\Delta^L_{am})^2}\Big) \,.
\end{align*}
\end{proof}

\end{appendices}

\begin{IEEEbiography}[{\includegraphics[width=1in,height=1.25in,clip,keepaspectratio]{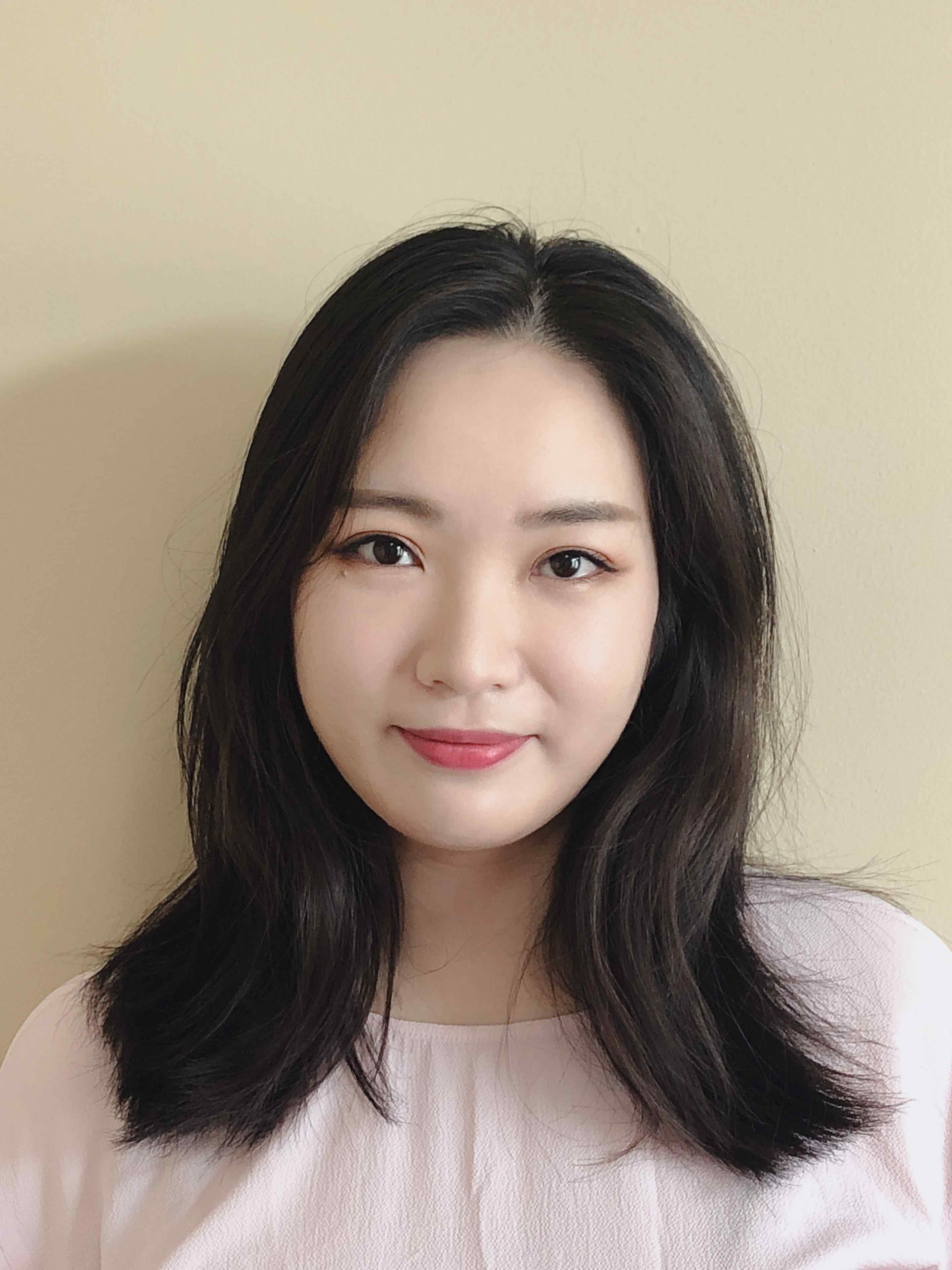}}]{Qinshuang Wei} received the B.S., M.S., and Ph.D. degrees in electrical engineering from Georgia Institute of Technology, Atlanta, Georgia, USA in 2017, 2018, and 2022, where she also received the B.S. degree in applied mathematics in 2017. She was a postdoctoral fellow in the University of Texas at Austin at 2022 and 2023. She is currently a postdoctoral research associate in the Purdue University.
Her research interests include cyber-physical systems with emphasis on multi-agent systems and transportation networks, optimization, and control systems. 
\end{IEEEbiography}

\begin{IEEEbiography}[{\includegraphics[width=1in,clip,keepaspectratio]{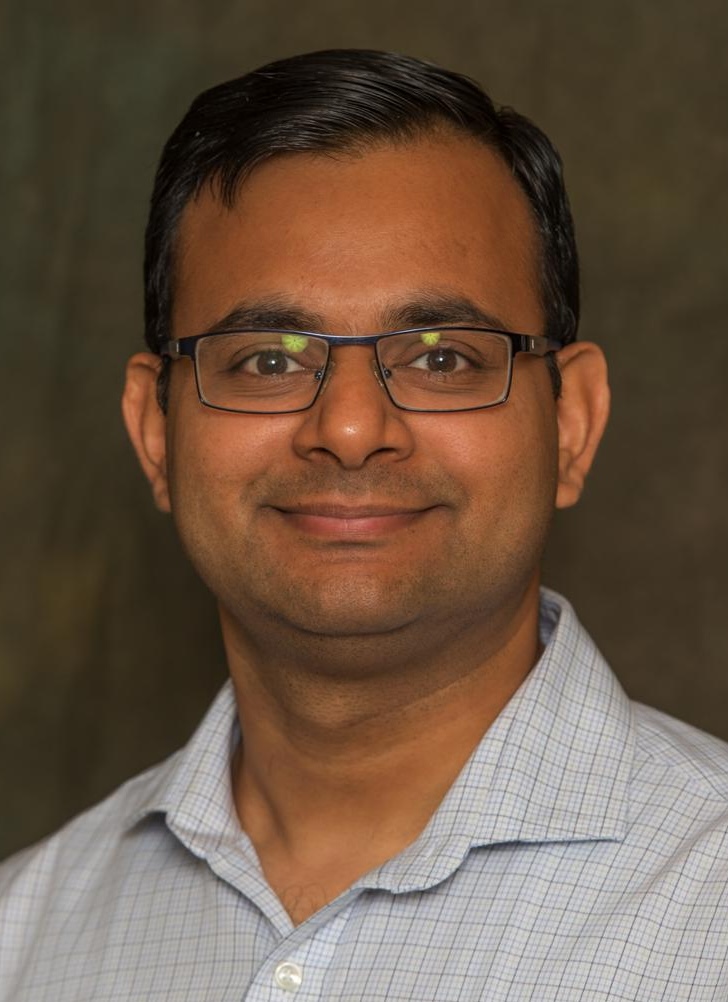}}]{Vaibhav Srivastava}~(S'08, M'14, SM'19) received the B.Tech. degree in mechanical engineering (2007) from the Indian Institute of Technology Bombay, Mumbai, India; the M.S. degree in mechanical engineering (2011), the M.A. degree in statistics (2012), and the Ph.D. degree in mechanical engineering (2012) from the University of California at Santa Barbara, Santa Barbara, CA. 	
He is currently an Associate Professor of Electrical and Computer Engineering at Michigan State University.  His research focuses on Cyber Physical Human Systems with emphasis on mixed human-robot systems, networked multi-agent systems, and autonomous vehicles.
 \end{IEEEbiography}

\begin{IEEEbiography}[{\includegraphics[width=1in,height=1.25in,clip,keepaspectratio]{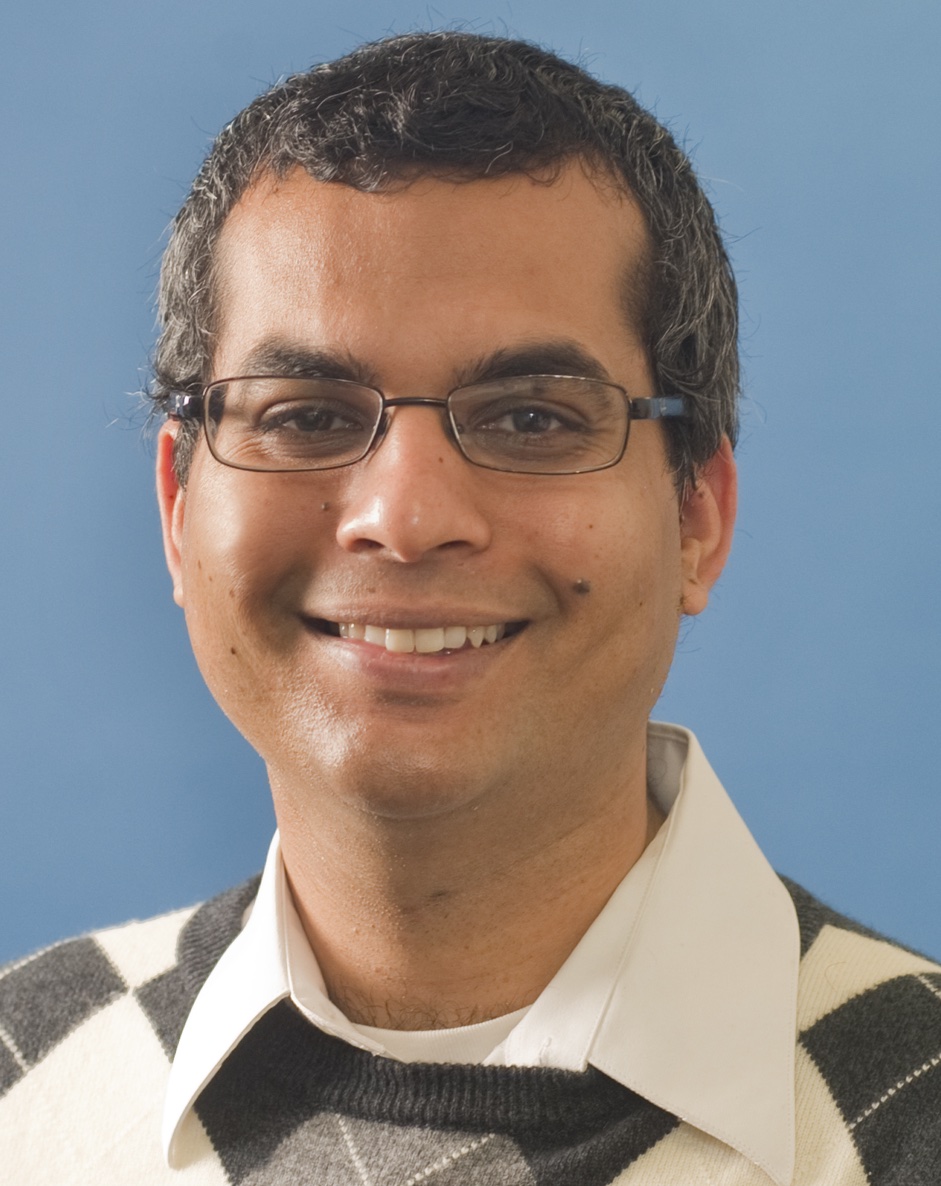}}]{Vijay Gupta} is in the Elmore Family School of Electrical and Computer Engineering at Purdue University. He received his B.~Tech degree at the Indian Institute of Technology, Delhi, and his M.S. and Ph.D. at California Institute of Technology, all in Electrical Engineering. He received the 2018 Antonio J Rubert Award from the IEEE Control Systems Society, the 2013 Donald P. Eckman Award from the American Automatic Control Council, and a 2009 National Science Foundation (NSF) CAREER Award. His research interests are broadly at the interface of communication, control, distributed computation, and human decision making.
\end{IEEEbiography}

\end{document}